\documentclass[reqno,12pt]{amsart}

\usepackage{amsmath}
\usepackage{latexsym}
\usepackage{amssymb}
\usepackage{hyperref}
\usepackage{graphicx}
\usepackage{epstopdf}
\usepackage{ifpdf}
\ifpdf
\fi

\makeatletter
 
 \newcommand{\Rmnum}[1]{\expandafter\@slowromancap\romannumeral #1@}
 \makeatother

\newtheorem{theorem}{Theorem}[section]

\newtheorem{proposition}[theorem]{Proposition}
\newtheorem{lemma}[theorem]{Lemma}

\newtheorem{assumption}[theorem]{Assumption}

\newcommand{\C}{{\mathbb C}}

\newcommand{\be}{\begin{equation}}
\newcommand{\ee}{\end{equation}}
\newcommand{\bea}{\begin{eqnarray}}
\newcommand{\eea}{\end{eqnarray}}
\newcommand{\ba}{\begin{array}}
\newcommand{\ea}{\end{array}}

\newcommand{\id}{\mathbb{I}}

\newcommand{\re}{\mathrm{Re}}

\newcommand{\eps}{\varepsilon}

\newcommand{\sig}{\sigma}

\newcommand{\lam}{\lambda}
\newcommand{\Lam}{\Lambda}
\newcommand{\gam}{\gamma}

\newcommand{\Om}{\Omega}

\newcommand{\dta}{\delta}
\newcommand{\Dta}{\Delta}

\newcommand{\tha}{\theta}

\newcommand{\pt}{\partial}

\numberwithin{equation}{section}

\begin{document}
\title[RHP for the Sasa-Satsuma equation on the half-line]{The Fokas method to the Sasa-Satsuma equation on the half-line}

\author[J.Xu]{Jian Xu}
\address{School of Mathematical Sciences\\
Fudan University\\
Shanghai 200433\\
People's  Republic of China}
\email{11110180024@fudan.edu.cn}

\author[E.Fan]{Engui Fan*}
\address{School of Mathematical Sciences, Institute of Mathematics and Key Laboratory of Mathematics for Nonlinear Science\\
Fudan University\\
Shanghai 200433\\
People's  Republic of China}
\email{correspondence author:faneg@fudan.edu.cn}

\keywords{Riemann-Hilbert problem, Sasa-Satsuma equation, Initial-boundary value problem}

\date{\today}

\begin{abstract}
We present a Riemann-Hilbert problem formalism for the initial-boundary value problem for the Sasa-Satsuma(SS) equation:
on the half-line. And we also analysis the global relation in this paper.

\end{abstract}

\maketitle

\section{Introduction}

Several of the most important PDEs in mathematics and physics are integrable. Integrable PDEs can be analyzed by means of the
Inverse Scattering Transform (IST) formalism. Until the 1990s the IST methodology was pursued almost entirely for pure initial value problems. However, in many laboratory and field situations, the wave motion is initiated by what corresponds to the imposition of boundary conditions rather than initial conditions. This naturally leads to the formulation of an initial-boundary value (IBV) problem instead of a pure initial value problem.
\par
In 1997, Fokas announced a new unified approach for the analysis of IBV problems for linear and nonlinear integrable PDEs \cite{f1,f2}(see also \cite{f3}). The Fokas method provides a generalization of the IST formalism from initial value to IBV problems, and over the last fifteen years, this method has been used to analyze boundary value problems for several of the most important integrable equations with $2\times 2$ Lax pairs, such as the Korteweg–de Vries, the nonlinear Schr\"o dinger, the sine-Gordon, and the stationary axisymmetric Einstein equations, see e.g. \cite{abmfs1,l2}. Just like the IST on the line, the unified method yields an expression for the solution of an IBV problem in terms of the solution of a Riemann-Hilbert problem. In particular, the asymptotic behavior of the solution can be analyzed in an effective way by using this Riemann-Hilbert problem and by employing the nonlinear version of the steepest descent method introduced by Deift and Zhou \cite{dz}.
\par
It is well known that the nonlinear Schr\"odinger(NLS) equation
\be \label{NLSe}
iq_T+\frac{1}{2}q_{XX}+|q|^2q=0
\ee
describes slowly varying wave envelopes in dispersive media and arises in various physical systems such as water waves, plasma physics, solid-state physics and nonlinear optics. One of the most successful among them is the description of optical solitons in fibers. But, by the advancement of experomenal accuracy, several phenomena which can not be explained by equation (\ref{NLSe}) have been observed. In order to understand such phenomena, Kodama and Hasegawa proposed a higer-order nonlinear Schr\"odinger equation
\be \label{HNLSe}
iq_T+\frac{1}{2}q_{XX}+|q|^2q+i\eps \{\beta_1q_{xxx}+\beta_2|q|^2q_{X}+\beta_3q(|q|^2)_X\}=0.
\ee
\par
In general, equation (\ref{HNLSe}) may not be completely integrable. However, if some restrictions are imposed on the real parameters $\beta_1,\beta_2$ and $\beta_3$, then we can apply the IST to solve its initial value problems. Until now, the following four cases besides the NLS equation itself are konwn to be solvable:
\begin{itemize}
\item the derivative NLS equation-type \Rmnum{1}($\beta_1:\beta_2:\beta_3$=0:1:1),
\item the derivative NLS equation-type \Rmnum{2}($\beta_1:\beta_2:\beta_3$=0:1:0),
\item the Hirota equation($\beta_1:\beta_2:\beta_3$=1:6:0),
\item the Sasa-Satsuma equation($\beta_1:\beta_2:\beta_3$=1:6:3).
       \be\label{SSe}
       iq_T+\frac{1}{2}q_{XX}+|q|^2q+i\eps (q_{XXX}+6|q|^2q_X+3q(|q|^2)_X)=0
       \ee
\end{itemize}
Recently, Lenells develop a methodology for analyzing IBV problems for integrable evolution equations with Lax pairs involving $3\times 3$ matrices \cite{l3}. He also used this method to analyze the Degasperis-Procesi equation in \cite{l4}. In this paper we analyze the initial-boundary value problem of the Sasa-Satsuma equation on the half-line by using this method. The IST formalism for the initial value problem of the Sasa-Satsuma equation has been obtained in \cite{ss}.
\par
According to \cite{ss} we introduce variable transformations,
\begin{subequations} \label{vartrans}
\be\label{uqtrans}
u(x,t)=q(X,T)\exp\{\frac{-i}{6\eps}(X-\frac{T}{18\eps})\},
\ee
\be
t=T,
\ee
\be
x=X-\frac{T}{12\eps}.
\ee
\end{subequations}
Then equation (\ref{HNLSe}) is reduce to a complex modified KdV-type equation
\be\label{KdV-typee}
u_t+\eps\{u_{xxx}+6|u|^2u_x+3u(|u|^2)_x\}=0.
\ee
\par
\par
{\bf Organization of the paper.}In section 2 we perform the spectral analysis of the associated Lax pair. And we formulate the main Riemann-Hilbert problem in section 3. We also analysis the global relation in section 4.

\section{Spectral analysis}
The Lax pair of equation (\ref{KdV-typee}) is \cite{ss},
\begin{subequations}
\be\label{Lax-x}
\Psi_x=U\Psi,\quad \Psi=\left(\ba{c}\Psi_1\\\Psi_2\\\Psi_3\ea\right).
\ee
\be\label{Lax-t}
\Psi_t=V\Psi.
\ee
\end{subequations}
where
\be\label{Udef}
U=-ik\Lam+V_1.
\ee
and
\be\label{Vdef}
V=-4i\eps k^3\Lam+V_2
\ee
here
\be\label{Lamdef}
\Lam=\left(\ba{ccc}1&0&0\\0&1&0\\0&0&-1\ea\right),V_1=\left(\ba{ccc}0&0&u\\0&0&\bar u\\-\bar u&-u&0\ea\right),V_2=k^2V_2^{(2)}+kV_2^{(1)}+V_2^{(0)}.
\ee
where
\be
\ba{l}
V_2^{(2)}=4\eps\left(\ba{ccc}0&0&u\\0&0&\bar u\\-\bar u&-u&0\ea\right),\\
V_2^{(1)}=2i\eps\left(\ba{ccc}|u|^2&u^2&u_x\\\bar u^2&|u|^2&\bar u_x\\\bar u_x&u_x&-2|u|^2\ea\right),\\
V_2^{(0)}=-4|u|^2\eps\left(\ba{ccc}0&0&u\\0&0&\bar u\\-\bar u&-u&0\ea\right)-\eps\left(\ba{ccc}0&0&u_{xx}\\0&0&\bar u_{xx}\\-\bar u_{xx}&-u_{xx}&0\ea\right)+\eps(u\bar u_x-u_x\bar u)\left(\ba{ccc}1&0&0\\0&-1&0\\0&0&0\ea\right)
\ea
\ee
In the following, we let $\eps=1$ for the convenient of the analysis.

\subsection{The closed one-form}
Suppose that $u(x,t)$ is sufficiently smooth function of $(x,t)$ in the half-line domain $\Om=\{0<x<\infty,0<t<T\}$ which decay as $x\rightarrow \infty$. Introducing a new eigenfunction $\mu(x,t,k)$ by
\be\label{neweigfun}
\Psi=\mu e^{-i\Lam kx-4i\Lam k^3t}
\ee
then we find the Lax pair equations
\be\label{muLaxe}
\left\{
\ba{l}
\mu_x+[ik\Lam,\mu]=V_1\mu,\\
\mu_t+[4ik^3\Lam,\mu]=V_2\mu.
\ea
\right.
\ee
the equations in (\ref{muLaxe}) can be written in differential form as
\be\label{mudiffform}
d(e^{(ikx+4ik^3t)\hat \Lam}\mu)=W,
\ee
where $W(x,t,k)$ is the closed one-form defined by
\be\label{Wdef}
W=e^{(ikx+4ik^3t)\hat \Lam}(V_1dx+V_2dt)\mu.
\ee

\subsection{The $\mu_j$'s}
We define three eigenfunctions $\{\mu_j\}_1^3$ of (\ref{muLaxe}) by the Volterra integral equations
\be\label{mujdef}
\mu_j(x,t,k)=\id+\int_{\gam_j}e^{(-i kx-4i k^3t)\hat \Lam}W_j(x',t',k).\qquad j=1,2,3.
\ee
where $W_j$ is given by (\ref{Wdef}) with $\mu$ replaced with $\mu_j$, and the contours $\{\gam_j\}_1^3$ are showed in Figure 1.
\begin{figure}[th]
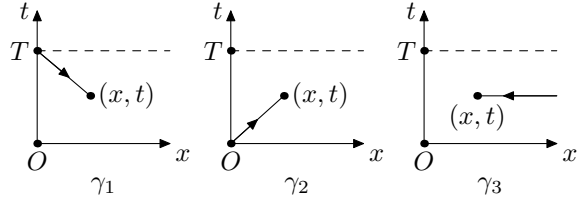

\centering
\includegraphics{SS-HL.1}
\includegraphics{SS-HL.2}
\includegraphics{SS-HL.3}
\caption{The three contours $\gam_1,\gam_2$ and $\gam_3$ in the $(x,t)-$domain.}
\end{figure}
The first, second and third column of the matrix equation (\ref{mujdef}) involves the exponentials
\be
\ba{ll}
\mbox{$[\mu_j]_1$:}&e^{2ik(x-x')+8ik^3(t-t')},\\
\mbox{$[\mu_j]_2$:}&e^{2ik(x-x')+8ik^3(t-t')},\\
\mbox{$[\mu_j]_3$:}&e^{-2ik(x-x')-8ik^3(t-t')},e^{-2ik(x-x')-8ik^3(t-t')}.
\ea
\ee
And we have the following inequalities on the contours:
\be
\ba{ll}
\gam_1:&x-x'\ge 0,t-t'\le 0,\\
\gam_2:&x-x'\ge 0,t-t'\ge 0,\\
\gam_3:&x-x'\le 0.
\ea
\ee
So, these inequalities imply that the functions $\{\mu_j\}_1^3$ are bounded and analytic for $k\in\C$ such that $k$ belongs to
\be\label{mujbodanydom}
\ba{ll}
\mu_1:&(D_2,D_2,D_3),\\
\mu_2:&(D_1,D_1,D_4),\\
\mu_3:&(D_3\cup D_4,D_3\cup D_4,D_1\cup D_2).
\ea
\ee
where $\{D_n\}_1^4$ denote four open, pairwisely disjoint subsets of the Riemann $k-$sphere showed in Figure 2.
\begin{figure}[th]
\centering
\includegraphics{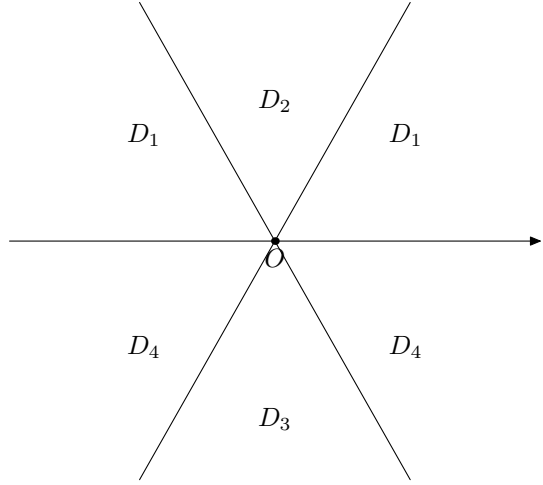}
\caption{The sets $D_n$, $n=1,\ldots ,4$, which decompose the complex $k-$plane.}
\end{figure}
And the sets $\{D_n\}_1^4$ has the following properties:
\[
\ba{l}
D_1=\{k\in\C|\re{l_1}=\re{l_2}>\re{l_3},\re{z_1}=\re{z_2}>\re{z_3}\},\\
D_2=\{k\in\C|\re{l_1}=\re{l_2}>\re{l_3},\re{z_1}=\re{z_2}<\re{z_3}\},\\
D_1=\{k\in\C|\re{l_1}=\re{l_2}<\re{l_3},\re{z_1}=\re{z_2}>\re{z_3}\},\\
D_1=\{k\in\C|\re{l_1}=\re{l_2}<\re{l_3},\re{z_1}=\re{z_2}<\re{z_3}\},\\
\ea
\]
where $l_i(k)$ and $z_i(k)$ are the diagonal entries of matrices $-ik\Lam$ and $-4ik^3\Lam$, respectively.
\par
In fact, for $x=0$, $\mu_1(0,t,k)$ has enlarged domain of boundedness: $(D_2\cup D_4,D_2\cup D_4,D_1\cup D_3)$, and $\mu_2(0,t,k)$ has enlarged domain of boundedness: $(D_1\cup D_3,D_1\cup D_3,D_2\cup D_4)$.

\subsection{The $M_n$'s}
For each $n=1,\ldots,4$, define a solution $M_n(x,t,k)$ of (\ref{muLaxe}) by the following system of integral equations:
\be\label{Mndef}
(M_n)_{ij}(x,t,k)=\dta_{ij}+\int_{\gam_{ij}^n}(e^{(-i kx-4i k^3t)\hat \Lam}W_n(x',t',k))_{ij},\quad k\in D_n,\quad i,j=1,2,3.
\ee
where $W_n$ is given by (\ref{Wdef}) with $\mu$ replaced with $M_n$, and the contours $\gam_{ij}^n$, $n=1,\ldots,4$, $i,j=1,2,3$ are defined by
\be\label{gamijndef}
\gam_{ij}^n=\left\{\ba{lclcl}\gam_1&if&\re l_i(k)<\re l_j(k)&and&\re z_i(k)\ge\re z_j(k),\\\gam_2&if&\re l_i(k)<\re l_j(k)&and&\re z_i(k)<\re z_j(k),\\\gam_3&if&\re l_i(k)\ge\re l_j(k)&&.\\\ea\right.\quad \mbox{for }\quad k\in D_n.
\ee

\par
The following proposition ascertains that the $M_n$'s defined in this way have the properties required for the formulation of a Riemann-Hilbert problem.
\begin{proposition}
For each $n=1,\ldots,4$, the function $M_n(x,t,k)$ is well-defined by equation (\ref{Mndef}) for $k\in \bar D_n$ and $(x,t)\in \Om$. For any fixed point $(x,t)$, $M_n$ is bounded and analytic as a function of $k\in D_n$ away from a possible discrete set of singularities $\{k_j\}$ at which the Fredholm determinant vanishes. Moreover, $M_n$ admits a bounded and contious extension to $\bar D_n$ and
\be\label{Mnasy}
M_n(x,t,k)=\id+O(\frac{1}{k}),\qquad k\rightarrow \infty,\quad k\in D_n.
\ee
\end{proposition}
\begin{proof}
The bounedness and analyticity properties are established in appendix B in \cite{l3}. And substituting the expansion
\[
M=M_0+\frac{M^{(1)}}{k}+\frac{M^{(2)}}{k^2}+\cdots,\qquad k\rightarrow \infty.
\]
into the Lax pair (\ref{muLaxe}) and comparing the terms of the same order of $k$ yield the equation (\ref{Mnasy}).
\end{proof}

\subsection{The jump matrices}

We define spectral functions $S_n(k)$, $n=1,\ldots,4$, and
\be\label{Sndef}
S_n(k)=M_n(0,0,k),\qquad k\in D_n,\quad n=1,\ldots,4.
\ee
Let $M$ denote the sectionally analytic function on the Riemann $k-$sphere which equals $M_n$ for $k\in D_n$. Then $M$ satisfies the jump conditions
\be\label{Mjump}
M_n=M_mJ_{m,n},\qquad k\in \bar D_n\cap \bar D_m,\qquad n,m=1,\ldots,4,\quad n\ne m,
\ee
where the jump matrices $J_{m,n}(x,t,k)$ are defined by
\be\label{Jmndef}
J_{m,n}=e^{(-i kx-4i k^3t)\hat \Lam}(S_m^{-1}S_n).
\ee
According to the definition of the $\gam^n$, we find that
\be\label{gamndef}
\ba{ll}
\gam^1=\left(\ba{lll}\gam_3&\gam_3&\gam_3\\\gam_3&\gam_3&\gam_3\\\gam_2&\gam_2&\gam_3\ea\right)&
\gam^2=\left(\ba{lll}\gam_3&\gam_3&\gam_3\\\gam_3&\gam_3&\gam_3\\\gam_1&\gam_1&\gam_3\ea\right)\\
\gam^3=\left(\ba{lll}\gam_3&\gam_3&\gam_1\\\gam_3&\gam_3&\gam_1\\\gam_3&\gam_3&\gam_3\ea\right)&
\gam^4=\left(\ba{lll}\gam_3&\gam_3&\gam_2\\\gam_3&\gam_3&\gam_2\\\gam_3&\gam_3&\gam_3\ea\right).
\ea
\ee

\subsection{The adjugated eigenfunctions}
We will also need the analyticity and boundedness properties of the minors of the matrices $\{\mu_j(x,t,k)\}_1^3$. We recall that the adjugate matrix $X^A$ of a $3\times 3$ matrix $X$ is defined by
\[
X^A=\left(
\ba{ccc}
m_{11}(X)&-m_{12}(X)&m_{13}(X)\\
-m_{21}(X)&m_{22}(X)&-m_{23}(X)\\
m_{31}(X)&-m_{32}(X)&m_{33}(X)
\ea
\right),
\]
where $m_{ij}(X)$ denote the $(ij)$th minor of $X$.
\par
It follows from (\ref{muLaxe}) that the adjugated eigenfunction $\mu^A$ satisfies the Lax pair
\be\label{muadgLaxe}
\left\{
\ba{l}
\mu_x^A-[ik\Lam,\mu^A]=-V_1^T\mu^A,\\
\mu_t^A-[4ik^3\Lam,\mu^A]=-V_2^T\mu^A.
\ea
\right.
\ee
where $V^T$ denote the transform of a matrix $V$.
Thus, the eigenfunctions $\{\mu_j^A\}_1^3$ are solutions of the integral equations
\be\label{muadgdef}
\mu_j^A(x,t,k)=\id-\int_{\gam_j}e^{ik(x-x')+4ik^3(t-t')\hat \Lam}(V_1^Tdx+V_2^T)\mu^A,\quad j=1,2,3.
\ee
Then we can get the following analyticity and boundedness properties:
\be\label{mujadgbodanydom}
\ba{ll}
\mu_1^A:&(D_3,D_3,D_2),\\
\mu_2^A:&(D_4,D_4,D_1),\\
\mu_3^A:&(D_1\cup D_2,D_1\cup D_2,D_3\cup D_4).
\ea
\ee
In fact, for $x=0$, $\mu_1^A(0,t,k)$ has enlarged domain of boundedness: $(D_1\cup D_3,D_1\cup D_3,D_2\cup D_4)$, and $\mu_2^A(0,t,k)$ has enlarged domain of boundedness: $(D_2\cup D_4,D_2\cup D_4,D_1\cup D_3)$.

\subsection{The $J_{m,n}$'s computation}

Let us define the $3\times 3-$matrix value spectral functions $s(k)$ and $S(k)$ by
\begin{subequations}\label{sSdef}
\be\label{mu3mu2s}
\mu_3(x,t,k)=\mu_2(x,t,k)e^{(-ikx-4ik^3t)\hat \Lam}s(k), 
\ee
\be\label{mu1mu2S}
\mu_1(x,t,k)=\mu_2(x,t,k)e^{(-ikx-4ik^3t)\hat \Lam}S(k), 
\ee
\end{subequations}
Thus,
\be\label{sSmu3mu1}
s(k)=\mu_3(0,0,k),\qquad S(k)=\mu_1(0,0,k).
\ee
And we deduce from the properties of $\mu_j$ and $\mu_j^A$ that $s(k)$ and $S(k)$ have the following boundedness properties:
\[
\ba{ll}
s(k):&(D_3\cup D_4,D_3\cup D_4,D_1\cup D_2),\\
S(k):&(D_2\cup D_4,D_2\cup D_4,D_1\cup D_3),\\
s^A(k):&(D_1\cup D_2,D_1\cup D_2,D_3\cup D_4),\\
S^A(k):&(D_1\cup D_3,D_1\cup D_3,D_2\cup D_4).
\ea
\]
Moreover,
\be\label{MnSnrel}
M_n(x,t,k)=\mu_2(x,t,k)e^{(-ikx-4ik^3t)\hat\Lam}S_n(k),\quad k\in D_n.
\ee
\begin{proposition}
The $S_n$ can be expressed in terms of the entries of $s(k)$ and $S(k)$ as follows:
\begin{subequations}\label{Sn}
\be
\ba{l}
S_1=\left(\ba{ccc}\frac{m_{22}(s)}{s_{33}}&\frac{m_{21}(s)}{s_{33}}&s_{13}\\\frac{m_{12}(s)}{s_{33}}&\frac{m_{11}(s)}{s_{33}}
&s_{23}\\0&0&s_{33}\ea\right),\\
S_2=\left(\ba{ccc}\frac{m_{22}(s)m_{33}(S)-m_{32}(s)m_{23}(S)}{(s^TS^A)_{33}}&\frac{m_{21}(s)m_{33}(S)-m_{31}(s)m_{23}(S)}{(s^TS^A)_{33}}&s_{13}\\
\frac{m_{12}(s)m_{33}(S)-m_{32}(s)m_{13}(S)}{(s^TS^A)_{33}}&\frac{m_{11}(s)m_{33}(S)-m_{31}(s)m_{13}(S)}{(s^TS^A)_{33}}&s_{23}\\
\frac{m_{12}(s)m_{23}(S)-m_{22}(s)m_{13}(S)}{(s^TS^A)_{33}}&\frac{m_{11}(s)m_{23}(S)-m_{21}(s)m_{13}(S)}{(s^TS^A)_{33}}&s_{33}\ea\right),\\
\ea
\ee
\be
\ba{ll}
S_3=\left(\ba{ccc}s_{11}&s_{12}&\frac{S_{13}}{(S^Ts^A)_{33}}\\s_{21}&s_{22}&\frac{S_{23}}{(S^Ts^A)_{33}}\\s_{31}&s_{32}&\frac{S_{33}}{(S^Ts^A)_{33}}\ea\right),&
S_4=\left(\ba{ccc}s_{11}&s_{12}&0\\s_{21}&s_{22}&0\\s_{31}&s_{32}&\frac{1}{m_{33}(s)}\ea\right).
\ea
\ee
\end{subequations}
\end{proposition}
\begin{proof}
Let $\gam_3^{X_0}$ denote the contour $(X_0,0)\rightarrow (x,t)$ in the $(x,t)-$plane, here $X_0>0$ is a constant. We introduce $\mu_3(x,t,k;X_0)$ as the solution of (\ref{mujdef}) with $j=3$ and with the contour $\gam_3$ replaced by $\gam_3^{X_0}$. Similarly, we define $M_n(x,t,k;X_0)$ as the solution of (\ref{Mndef}) with $\gam_3$ replaced by $\gam_3^{X_0}$. We will first derive expression for $S_n(k;X_0)=M_n(0,0,k;X_0)$ in terms of $S(k)$ and $s(k;X_0)=\mu_3(0,0,k;X_0)$. Then (\ref{Sn}) will follow by taking the limit $X_0\rightarrow \infty$.
\par
First, We have the following relations:
\be\label{MnRnSnTn}
\left\{
\ba{l}
M_n(x,t,k;X_0)=\mu_1(x,t,k)e^{(-ikx-4ik^3t)\hat \Lam}R_n(k;X_0),\\
M_n(x,t,k;X_0)=\mu_2(x,t,k)e^{(-ikx-4ik^3t)\hat \Lam}S_n(k;X_0),\\
M_n(x,t,k;X_0)=\mu_3(x,t,k)e^{(-ikx-4ik^3t)\hat \Lam}T_n(k;X_0).
\ea
\right.
\ee
Then we get $R_n(k;X_0)$ and $T_n(k;X_0)$ are fedined as follows:
\begin{subequations}\label{RnTnX0}
\be\label{RnX0}
R_n(k;X_0)=e^{4ik^3T\hat \Lam}M_n(0,T,k;X_0),
\ee
\be\label{TnX0}
T_n(k;X_0)=e^{ikx\hat \Lam}M_n(X_0,0,k;X_0).
\ee
\end{subequations}
The relations (\ref{MnRnSnTn}) imply that
\be\label{sSRnSnTn}
s(k;X_0)=S_n(k;X_0)T^{-1}_n(k;X_0),\qquad S(k)=S_n(k;X_0)R^{-1}_n(k;X_0).
\ee
These equations constitute a matrix factorization problem which, given $\{s,S\}$ can be solved for the $\{R_n,S_n,T_n\}$. Indeed, the integral equations (\ref{Mndef}) together with the definitions of $\{R_n,S_n,T_n\}$ imply that
\be
\left\{
\ba{lll}
(R_n(k;X_0))_{ij}=0&if&\gam_{ij}^n=\gam_1,\\
(S_n(k;X_0))_{ij}=0&if&\gam_{ij}^n=\gam_2,\\
(T_n(k;X_0))_{ij}=0&if&\gam_{ij}^n=\gam_3.
\ea
\right.
\ee
It follows that (\ref{sSRnSnTn}) are 18 scalar equations for 18 unknowns. By computing the explicit solution of this algebraic system, we find that $\{S_n(k;X_0)\}_1^4$ are given by the equation obtained from (\ref{Sn}) by replacing $\{S_n(k),s(k)\}$ with $\{S_n(k;X_0),s(k;X_0)\}$. taking $X_0\rightarrow \infty$ in this equation, we arrive at (\ref{Sn}).
\end{proof}

\subsection{The global relation}
The spectral functions $S(k)$ and $s(k)$ are not independent but satisfy an important relation. Indeed, it follows from (\ref{sSdef}) that
\be
\mu_1(x,t,k)e^{(-ikx-4ik^3t)\hat \Lam}S^{-1}(k)s(k)=\mu_3(x,t,k),\quad k\in(D_3\cup D_4,D_3\cup D_4,D_1\cup D_2).
\ee
Since $\mu_1(0,T,k)=\id$, evaluation at $(0,T)$ yields the following global relation:
\be\label{globalrel}
S^{-1}(k)s(k)=e^{4ik^3T\hat \Lam}c(T,k),\quad k\in(D_3\cup D_4,D_3\cup D_4,D_1\cup D_2).
\ee
where $c(T,k)=\mu_3(0,T,k)$.

\subsection{The residue conditions}
Since $\mu_2$ is an entire function, it follows from (\ref{MnSnrel}) that M can only have sigularities at the points where the $S_n's$ have singularities. We infer from the explicit formulas (\ref{Sn}) that the possible singularities of $M$ are as follows:
\begin{itemize}
\item $[M]_1$ could have poles in $D_1\cup D_2$ at the zeros of $s_{33}(k)$;
\item $[M]_1$ could have poles in $D_2$ at the zeros of $(s^TS^A)_{33}(k)$;
\item $[M]_2$ could have poles in $D_1\cup D_2$ at the zeros of $s_{33}(k)$;
\item $[M]_2$ could have poles in $D_2$ at the zeros of $(s^TS^A)_{33}(k)$;
\item $[M]_3$ could have poles in $D_3$ at the zeros of $(S^Ts^A)_{33}(k)$;
\item $[M]_3$ could have poles in $D_3\cup D_4$ at the zeros of $m_{33}(s)(k)$;
\end{itemize}
We denote the above possible zeros by $\{k_j\}_1^N$ and assume they satisfy the following assumption.
\begin{assumption}
We assume that
\begin{itemize}
\item $s_{33}(k)$ has $n_0$ possible simple zeros in $D_1$ denoted by $\{k_j\}_1^{n_0}$;
\item $s_{33}(k)$ has $n_1-n_0$ possible simple zeros in $D_2$ denoted by $\{k_j\}_{n_0+1}^{n_1}$;
\item $(s^TS^A)_{33}(k)$ has $n_2-n_1$ possible simple zeros in $D_2$ denoted by $\{k_j\}_{n_1+1}^{n_2}$;
\item $(S^Ts^A)_{33}(k)$ has $n_3-n_2$ possible simple zeros in $D_3$ denoted by $\{k_j\}_{n_2+1}^{n_3}$;
\item $m_{33}(s)(k)$ has $n_4-n_3$ possible simple zeros in $D_3$ denoted by $\{k_j\}_{n_3+1}^{n_4}$;
\item $m_{33}(s)(k)$ has $n_5-n_4$ possible simple zeros in $D_3$ denoted by $\{k_j\}_{n_4+1}^{n_5}$;
\item $m_{33}(s)(k)$ has $N-n_5$ possible simple zeros in $D_4$ denoted by $\{k_j\}_{n_5+1}^{N}$;
\end{itemize}
and that none of these zeros coincide. Moreover, we assume that none of these functions have zeros on the boundaries of the $D_n$'s.
\end{assumption}
We determine the residue conditions at these zeros in the following:
\begin{proposition}\label{propos}
Let $\{M_n\}_1^4$ be the eigenfunctions defined by (\ref{Mndef}) and assume that the set $\{k_j\}_1^N$ of singularitues are as the above assumption. Then the following residue conditions hold:
\begin{subequations}
\be\label{M11D1res}
{Res}_{k=k_j}[M]_1=\frac{m_{12}(s)(k_j)}{\dot s_{33}(k_j)s_{23}(k_j)}e^{\tha_{31}(k_j)}[M(k_j)]_3,\quad 1\le j\le n_0,k_j\in D_1
\ee
\be\label{M12D1res}
{Res}_{k=k_j}[M]_2=\frac{m_{12}(s)(k_j)}{\dot s_{33}(k_j)s_{13}(k_j)}e^{\tha_{32}(k_j)}[M(k_j)]_3,\quad 1\le j\le n_0,k_j\in D_1
\ee
\be\label{M21D2res}
\ba{r}
Res_{k=k_j}[M]_1=\frac{m_{12}(s)(k_j)m_{33}(S)(k_j)-m_{32}(s)(k_j)m_{13}(S)(k_j)}{\dot{(s^TS^A)_{33}(k_j)}s_{23}(k_j)}e^{\tha_{31}(k_j)}[M(k_j)]_3\\
\quad n_1+1\le j\le n_2,k_j\in D_2,
\ea
\ee
\be\label{M22D2res}
\ba{r}
Res_{k=k_j}[M]_2=\frac{m_{21}(s)(k_j)m_{33}(S)(k_j)-m_{31}(s)(k_j)m_{23}(S)(k_j)}{\dot{(s^TS^A)_{33}(k_j)}s_{13}(k_j)}e^{\tha_{32}(k_j)}[M(k_j)]_3\\
\quad n_1+1\le j\le n_2,k_j\in D_2,
\ea
\ee
\be\label{M43D4res}
\ba{rl}
Res_{k=k_j}[M]_3=&\frac{S_{13}(k_j)s_{32}(k_j)-S_{33}(k_j)s_{12}(k_j)}{\dot{(S^Ts^A)_{33}(k_j)}m_{23}(s)(k_j)}e^{\tha_{13}(k_j)}[M(k_j)]_1\\
&+\frac{S_{33}(k_j)s_{11}(k_j)-S_{13}(k_j)s_{31}(k_j)}{\dot{(S^Ts^A)_{33}(k_j)}m_{23}(s)(k_j)}e^{\tha_{23}(k_j)}[M(k_j)]_2,n_2+1\le j\le n_3,k_j\in D_3,
\ea
\ee
\be\label{M63D6res}
\ba{r}
Res_{k=k_j}[M]_3=\frac{s_{12}(k_j)}{\dot m_{33}(s)(k_j)m_{23}(s)(k_j)}e^{\tha_{13}(k_j)}[M(k_j)]_1-\frac{s_{11}(k_j)}{\dot m_{33}(s)(k_j)m_{23}(s)(k_j)}e^{\tha_{23}(k_j)}[M(k_j)]_2\\
\quad n_4+1\le j\le N,k_j\in D_4.
\ea
\ee
\end{subequations}
where $\dot f=\frac{df}{dk}$, and $\tha_{ij}$ is defined by
\be\label{thaijdef}
\tha_{ij}(x,t,k)=(l_i-l_j)x+(z_i-z_j)t,\quad i,j=1,2,3.
\ee
that implies that
\[
\tha_{ij}=0,i,j=1,2;\quad \tha_{13}=\tha_{23}=-\tha_{32}=-\tha_{31}=-2ikx-8ik^3t.
\]
\end{proposition}
\begin{proof}
We will prove (\ref{M11D1res}), (\ref{M21D2res}), (\ref{M43D4res}), (\ref{M63D6res}), the other conditions follow by similar arguments.
Equation (\ref{MnSnrel}) implies the relation
\begin{subequations}
\be\label{M1S1}
M_1=\mu_2e^{(-ikx-4ik^3t)\hat\Lam}S_1,
\ee
\be\label{M2S2}
M_2=\mu_2e^{(-ikx-4ik^3t)\hat\Lam}S_2.
\ee
\be\label{M4S4}
M_3=\mu_2e^{(-ikx-4ik^3t)\hat\Lam}S_3,
\ee
\be\label{M6S6}
M_4=\mu_2e^{(-ikx-4ik^3t)\hat\Lam}S_4,
\ee
\end{subequations}
In view of the expressions for $S_1$ and $S_2$ given in (\ref{Sn}), the three columns of (\ref{M1S1}) read:
\begin{subequations}
\be\label{M11}
[M_1]_1=[\mu_2]_1\frac{m_{22}(s)}{s_{33}}+[\mu_2]_2e^{\tha_{21}}\frac{m_{12}(s)}{s_{33}},
\ee
\be\label{M12}
[M_1]_2=[\mu_2]_1e^{\tha_{12}}\frac{m_{21}(s)}{s_{33}}+[\mu_2]_2\frac{m_{11}(s)}{s_{33}},
\ee
\be\label{M13}
[M_1]_3=[\mu_2]_1e^{\tha_{13}}s_{13}+[\mu_2]_2e^{\tha_{23}}s_{23}+[\mu_2]_3s_{33}.
\ee
\end{subequations}
while the three columns of (\ref{M2S2}) read:
\begin{subequations}
\be\label{M21}
\ba{rl}
[M_2]_1&=[\mu_2]_1\frac{m_{22}(s)m_{33}(S)-m_{32}(s)m_{23}(S)}{(s^TS^A)_{33}}\\
&+[\mu_2]_2\frac{m_{12}(s)m_{33}(S)-m_{32}(s)m_{13}(S)}{(s^TS^A)_{33}}e^{\tha_{21}}\\
&+[\mu_2]_3\frac{m_{12}(s)m_{23}(S)-m_{22}(s)m_{13}(S)}{(s^TS^A)_{33}}e^{\tha_{31}}
\ea
\ee
\be\label{M22}
\ba{rl}
[M_2]_2&=[\mu_2]_1\frac{m_{21}(s)m_{33}(S)-m_{31}(s)m_{23}(S)}{(s^TS^A)_{33}}e^{\tha_{12}}\\
&+[\mu_2]_2\frac{m_{11}(s)m_{33}(S)-m_{31}(s)m_{13}(S)}{(s^TS^A)_{33}}\\
&+[\mu_2]_3\frac{m_{11}(s)m_{23}(S)-m_{21}(s)m_{13}(S)}{(s^TS^A)_{33}}e^{\tha_{32}}
\ea
\ee
\be\label{M23}
[M_2]_3=[\mu_2]_1s_{13}e^{\tha_{13}}+[\mu_2]_2s_{23}e^{\tha_{23}}+[\mu_2]_3s_{33}.
\ee
\end{subequations}
and the three columns of (\ref{M4S4}) read:
\begin{subequations}
\be\label{M41}
[M_3]_1=[\mu_2]_1s_{11}+[\mu_2]_2s_{21}e^{\tha_{21}}+[\mu_2]_3s_{31}e^{\tha_{31}},
\ee
\be\label{M42}
[M_3]_2=[\mu_2]_1s_{12}e^{\tha_{12}}+[\mu_2]_2s_{22}+[\mu_2]_3s_{32}e^{\tha_{32}},
\ee
\be\label{M43}
[M_3]_3=[\mu_2]_1\frac{S_{13}}{(S^Ts^A)_{33}}e^{\tha_{13}}+[\mu_2]_2\frac{S_{23}}{(S^Ts^A)_{33}}e^{\tha_{23}}+[\mu_2]_3\frac{S_{33}}{(S^Ts^A)_{33}}.
\ee
\end{subequations}
the three columns of (\ref{M6S6}) read:
\begin{subequations}
\be\label{M61}
[M_4]_1=[\mu_2]_1s_{11}+[\mu_2]_2s_{21}e^{\tha_{21}}+[\mu_2]_3s_{31}e^{\tha_{31}},
\ee
\be\label{M62}
[M_4]_2=[\mu_2]_1s_{12}e^{\tha_{12}}+[\mu_2]_2s_{22}+[\mu_2]_3s_{32}e^{\tha_{32}},
\ee
\be\label{M63}
[M_4]_3=[\mu_2]_3\frac{1}{m_{33}(s)}.
\ee
\end{subequations}
We first suppose that $k_j\in D_1$ is a simple zero of $s_{33}(k)$. Solving (\ref{M13}) for $[\mu_2]_2$ and substituting the result in to (\ref{M11}), we find
\[
[M_1]_1=\frac{m_{12}(s)}{s_{33}s_{23}}e^{\tha_{31}}[M_1]_3+\frac{m_{32}(s)}{s_{23}}[\mu_2]_2-\frac{m_{12}(s)}{s_{23}}e^{\tha_{31}}[\mu_2]_3.
\]
Taking the residue of this equation at $k_j$, we find the condition (\ref{M11D1res}) in the case when $k_j\in D_1$. Similarly, Solving (\ref{M23}) for $[\mu_2]_2$ and substituting the result in to (\ref{M21}), we find
\[
[M_2]_1=\frac{m_{12}(s)m_{33}(S)-m_{32}(s)m_{13}(S)}{(s^TS^A)_{33}s_{23}}e^{\tha_{31}}[M_1]_3-\frac{m_{32}(s)}{s_{23}}[\mu_2]_1-\frac{m_{12}(s)}{s_{23}}e^{\tha_{31}}[\mu_2]_3.
\]
Taking the residue of this equation at $k_j$, we find the condition (\ref{M21D2res}) in the case when $k_j\in D_2$.
\par
In order to prove (\ref{M43D4res}), we solve (\ref{M41}) and (\ref{M42}) for $[\mu_2]_1$ and $[\mu_2]_3$, then substituting the result into (\ref{M43}), we find
\[
[M_3]_3=\frac{S_{13}s_{32}-S_{33}s_{12}}{(S^Ts^A)_{33}m_{23}(s)}e^{\tha_{13}}[M_3]_1+\frac{S_{33}s_{11}-S_{13}s_{31}}{(S^Ts^A)_{33}(k_j)m_{23}(s)}e^{\tha_{23}}[M_3]_2+\frac{1}{m_{23}(s)}[\mu_2]_3.
\]
Taking the residue of this equation at $k_j$, we find the condition (\ref{M43D4res}) in the case when $k_j\in D_3$. Similarly, solving (\ref{M61}) and (\ref{M62}) for $[\mu_2]_1$ and $[\mu_2]_3$, then substituting the result into (\ref{M63}), we find
\[
[M_4]_3=\frac{s_{12}}{m_{33}(s)m_{23}(s)}e^{\tha_{13}}[M_4]_1-\frac{s_{11}}{m_{33}(s)m_{23}(s)}e^{\tha_{13}}[M_4]_2-\frac{1}{m_{23}(s)}e^{\tha_{23}}[\mu_2]_2.
\]
Taking the residue of this equation at $k_j$, we find the condition (\ref{M63D6res}) in the case when $k_j\in D_4$.
\end{proof}

\section{The Riemann-Hilbert problem}

The sectionally analytic function $M(x,t,k)$ defined in section 2 satisfies a Riemann-Hilbert problem which can be formulated in terms of the initial and boundary values of $u(x,t)$. By solving this Riemann-Hilbert problem, the solution of (\ref{KdV-typee})(then (\ref{SSe})) can be recovered for all values of $x,t$.
\begin{theorem}
Suppose that $u(x,t)$ is a solution of (\ref{KdV-typee}) in the half-line domain $\Om$ with sufficient smoothness and decays as $x\rightarrow \infty$. Then $u(x,t)$ can be reconstructed from the initial value $\{u_0(x)\}$ and boundary values $\{g_0(t),g_1(t),g_2(t)\}$ defined as follows,
\be\label{inibouvalu}
u_0(x)=u(x,0),\quad g_0(t)=u(0,t),\quad g_1(t)=u_x(0,t),\quad g_2(t)=u_{xx}(0,t).
\ee
\par
Use the initial and boundary data to define the jump matrices $J_{m,n}(x,t,k)$ as well as the spectral $s(k)$ and $S(k)$ by equation (\ref{sSdef}). Assume that the possible zeros $\{k_j\}_1^N$ of the functions $s_{33}(k),(s^TS^A)_{33}(k),(S^Ts^A)_{33}(k)$ and $m_{33}(s)(k)$ are as in assumption 2.3.
\par
Then the solution $\{u(x,t)\}$ is given by
\be\label{usolRHP}
u(x,t)=2i\lim_{k\rightarrow \infty}(kM(x,t,k))_{13}.
\ee
where $M(x,t,k)$ satisfies the following $3\times 3$ matrix Riemann-Hilbert problem:
\begin{itemize}
\item $M$ is sectionally meromorphic on the Riemann $k-$sphere with jumps across the contours $\bar D_n\cap \bar D_m,n,m=1,\cdots, 4$, see Figure 2.
\item Across the contours $\bar D_n\cap \bar D_m$, $M$ satisfies the jump condition
      \be\label{MRHP}
      M_n(x,t,k)=M_m(x,t,k)J_{m,n}(x,t,k),\quad k\in \bar D_n\cap \bar D_m,n,m=1,2,3,4.
      \ee
\item $M(x,t,k)=\id+O(\frac{1}{k}),\qquad k\rightarrow \infty$.
\item The residue condition of $M$ is showed in Proposition \ref{propos}.
\end{itemize}
\end{theorem}
\begin{proof}
It only remains to prove (\ref{usolRHP}) and this equation follows from the large $k$ asymptotics of the eigenfunctions, see the appendix \ref{apdix1}.
\end{proof}

\section{Non-linearizable Boundary Conditions}
A major difficulty of initial-boundary value problems is that some of the boundary values are unkown for a well-posed problem. All boundary values are needed for the definition of $S(k)$, and hence for the formulation of the Riemann-Hilbert problem. Our main result expresses the spectral function $S(k)$ in terms of the prescribed boundary data and the initial data via the solution of a system of nonlinear integral equations.

\subsection{Asymptotics}
An analysis of (\ref{muLaxe}) shows that the eigenfunctions $\{\mu_j\}_1^3$ have the following asymptotics as $k\rightarrow\infty$ (see the appendix \ref{apdix1}):
\begin{subequations}
\be\label{mujasykinf}
\ba{l}
\mu_j(x,t,k)=\id+\frac{1}{k}\left(\ba{lll}\frac{i}{2} \int_{(x_j,t_j)}^{(x,t)}\Dta&\frac{i}{2}\int_{(x_j,t_j)}^{(x,t)}\Dta_1&\frac{1}{2i}u\\\frac{i}{2} \int_{(x_j,t_j)}^{(x,t)}\Dta_2&\frac{i}{2} \int_{(x_j,t_j)}^{(x,t)}\Dta&\frac{1}{2i}\bar u\\\frac{1}{2i}\bar u&\frac{1}{2i}u&-i \int_{(x_j,t_j)}^{(x,t)}\Dta \ea\right)\\
+\frac{1}{k^2}\left(\ba{lll}-\frac{1}{4}\int_{(x_j,t_j)}^{(x,t)}(\eta+\nu_1)&-\frac{1}{4}\int_{(x_j,t_j)}^{(x,t)}\eta_1&\mu^{(2)}_{13}\\
-\frac{1}{4}\int_{(x_j,t_j)}^{(x,t)}\eta_2&-\frac{1}{4}\int_{(x_j,t_j)}^{(x,t)}(\eta+\nu_2)&\mu^{(2)}_{23}\\
\mu^{(2)}_{31}&\mu^{(2)}_{32}&\int_{(x_j,t_j)}^{(x,t)}\eta_3\\\ea\right)\\
+\frac{1}{k^3}\left(\ba{lll}\mu^{(3)}_{11}&\mu^{(3)}_{12}&\mu^{(3)}_{13}\\\mu^{(3)}_{21}&\mu^{(3)}_{22}&\mu^{(3)}_{23}\\\mu^{(3)}_{31}&\mu^{(3)}_{32}&\mu^{(3)}_{33}\ea\right)+O(\frac{1}{k^4})
\ea
\ee
\end{subequations}
where
\begin{subequations}
\be\label{Dtadef}
\ba{l}
\Dta=-|u|^2dx+(u\bar u_{xx}+u_{xx}\bar u-u_x\bar u_x+6|u|^4)dt\\
\Dta_1=-u^2dx+(uu_{xx}+u_{xx}u-(u_x)^2+6|u|^2u^2)dt\\
\Dta_2=-\bar u^2dx+(\bar u\bar u_{xx}+\bar u_{xx}\bar u-(\bar u_x)^2+6|u|^2\bar u^2)dt
\ea
\ee
\be
\ba{l}
\mu^{(2)}_{13}=-\frac{1}{2}u\int_{(x_j,t_j)}^{(x,t)}\Dta+\frac{1}{4}u_x\\
\mu^{(2)}_{23}=-\frac{1}{2}\bar u\int_{(x_j,t_j)}^{(x,t)}\Dta+\frac{1}{4}\bar u_x\\
\mu^{(2)}_{31}=\frac{1}{4}(\bar u\int_{(x_j,t_j)}^{(x,t)}\Dta+u\int_{(x_j,t_j)}^{(x,t)}\Dta_2)-\frac{1}{4}\bar u_x \\
\mu^{(2)}_{32}=\frac{1}{4}(\bar u\int_{(x_j,t_j)}^{(x,t)}\Dta_1+u\int_{(x_j,t_j)}^{(x,t)}\Dta)-\frac{1}{4}u_x.
\ea
\ee
\be
\ba{l}
\eta=d[\frac{1}{2}(\int_{(x_j,t_j)}^{(x,t)}\Dta)^2]\\
\nu_1=\Dta_1\int_{(x_j,t_j)}^{(x,t)}\Dta_2+(u\bar u_x)dx+(u\bar u_t-2|u|^2(u\bar u_x-u_x\bar u)-(u_{xx}\bar u_x-u_x\bar u_{xx}))dt\\
\eta_1=\int_{(x_j,t_j)}^{(x,t)}\Dta\int_{(x_j,t_j)}^{(x,t)}\Dta_1+u^2\\
\eta_2=\int_{(x_j,t_j)}^{(x,t)}\Dta\int_{(x_j,t_j)}^{(x,t)}\Dta_2+\bar u^2\\
\nu_2=\Dta_2\int_{(x_j,t_j)}^{(x,t)}\Dta_1+(\bar u u_x)dx+(\bar u u_t-2|u|^2(\bar u u_x-\bar u_x u)-(\bar u_{xx}u_x-\bar u_xu_{xx}))dt,\\
\eta_3=d[-\frac{1}{2}(\int_{(x_j,t_j)}^{(x,t)}\Dta)^2-\frac{1}{4}|u|^2].
\ea
\ee
and in the following we just use $\mu^{(3)}_{13},\mu^{(3)}_{23},\mu^{(3)}_{31}$ and $\mu^{(3)}_{32}$, so we only compute these functions
\be
\ba{l}
\mu^{(3)}_{13}=\frac{1}{2i}u\mu^{(2)}_{33}+\frac{1}{4}u_x\mu^{(1)}_{33}+\frac{i}{4}|u|^2u+\frac{i}{8}u_{xx}\\
\mu^{(3)}_{23}=\frac{1}{2i}\bar u\mu^{(2)}_{33}+\frac{1}{4}\bar u_x\mu^{(1)}_{33}+\frac{i}{4}|u|^2\bar u+\frac{i}{8}\bar u_{xx}\\
\mu^{(3)}_{31}=\frac{1}{2i}(\bar u\mu^{(2)_{11}}+u\mu^{(2)}_{21})-\frac{1}{4}(\bar u_x\mu^{(1)_{11}}+u_x\mu^{(1)}_{21})+\frac{i}{4}|u|^2\bar u+\frac{i}{8}\bar u_{xx}\\
\mu^{(3)}_{32}=\frac{1}{2i}(\bar u\mu^{(2)_{12}}+u\mu^{(2)}_{22})-\frac{1}{4}(\bar u_x\mu^{(1)_{12}}+u_x\mu^{(1)}_{22})+\frac{i}{4}|u|^2 u+\frac{i}{8}u_{xx}
\ea
\ee
\end{subequations}
From the global relation (\ref{globalrel})and replacing $T$ by $t$, we find
\be
\mu_2(0,t,k)e^{-4ik^3t\hat \Lam}s(k)=c(t,k),\quad k\in(D_3\cup D_4,D_3\cup D_4,D_1\cup D_2).
\ee
We define functions $\{\Phi_{13}(t,k),\Phi_{23}(t,k),\Phi_{33}(t,k)\}$ and $\{c_j(t,k)\}_1^3$ by:
\be
\mu_2(0,t,k)=\left(\ba{lll}\Phi_{11}(t,k)&\Phi_{12}(t,k)&\Phi_{13}(t,k)\\\Phi_{21}(t,k)&\Phi_{22}(t,k)&\Phi_{23}(t,k)\\\Phi_{31}(t,k)&\Phi_{32}(t,k)&\Phi_{33}(t,k)\ea\right),\quad \frac{[c(t,k)]_3}{s_{33}(k)}=\left(\ba{l}c_1(t,k)\\c_2(t,k)\\c_3(t,k)\ea\right).
\ee
we can write the $(13)$ and $(23)$ entries of the global relation as
\begin{subequations}\label{globalrelj3}
\be\label{globalrel13}
\Phi_{11}(t,k)e^{-8ik^3t}\frac{s_{13}}{s_{33}}+\Phi_{12}(t,k)e^{-8ik^3t}\frac{s_{23}}{s_{33}}+\Phi_{13}(t,k)=c_1(t,k),\quad k\in D_1\cup D_2,
\ee
\be\label{globalrel23}
\Phi_{21}(t,k)e^{-8ik^3t}\frac{s_{13}}{s_{33}}+\Phi_{22}(t,k)e^{-8ik^3t}\frac{s_{23}}{s_{33}}+\Phi_{23}(t,k)=c_2(t,k),\quad k\in D_1\cup D_2,
\ee
\end{subequations}
The functions $\{c_j(t,k)\}_1^3$ are analytic and bounded in $D_1\cup D_2$ away from the possible zeros of $s_{33}(k)$ and of order $O(\frac{1}{k})$ as $k\rightarrow \infty$.
\par
From the asymptotic of $\mu_j(x,t,k)$ in (\ref{mujasykinf}) we have
\be\label{sk3}
\left(\ba{l}s_{13}(k)\\s_{23}(k)\\s_{33}(k)\ea\right)=\left(\ba{l}0\\0\\1\ea\right)+\frac{1}{2ik}\left(\ba{l}u(0,0)\\\bar u(0,0)\\2\int_{(\infty,0)}^{(0,0)}\Dta\ea\right)+O(\frac{1}{k^2}).
\ee
and
\begin{subequations}\label{Phi3}
\be\label{Phi13}
\Phi_{j3}(t,k)=\frac{\Phi_{j3}^{(1)}(t)}{k}+\frac{\Phi_{j3}^{(2)}(t)}{k^2}+\frac{\Phi_{j3}^{(3)}(t)}{k^3}+O(\frac{1}{k^4}),
\ee
\be\label{Phi33}
\Phi_{33}(t,k)=1+\frac{\Phi_{33}^{(1)}(t)}{k}+\frac{\Phi_{33}^{(2)}(t)}{k^2}+O(\frac{1}{k^3}),\quad k\rightarrow \infty,k\in D_1\cup D_2.
\ee
where
\[
\ba{ll}
\Phi_{j3}^{(1)}(t)=\frac{1}{2i}g_0(t)^T,&\Phi_{j3}^{(2)}(t)=\frac{1}{4}g_1(t)^T-\frac{1}{2}g_0^T\int_{(0,0)}^{(x,t)}\Dta\\
\Phi_{j3}^{(3)}(t)=\frac{1}{2i}g_0^T\Phi^{(2)}_{33}+\frac{1}{4}g_1^T\Phi^{(1)}_{33}+\frac{i}{4}|u|^2g_0^T+\frac{i}{8}g_2^T,&\\
\Phi_{33}^{(1)}(t)=-i\int_{(0,0)}^{(x,t)}\Dta,&\Phi_{33}^{(2)}(t)=\int_{(x_j,t_j)}^{(x,t)}\eta_3.
\ea
\]
Here the definition of $\Phi_{j3}(t,k)$ can be found in the appendix \ref{apdix1}.
\end{subequations}
\par
In particular, we find the following expressions for the boudary values:
\begin{subequations}\label{g012}
\be\label{g0}
g_0^T=2i\Phi_{j3}^{(1)}(t),
\ee
\be\label{g1}
g_1^T=2ig_0^T\Phi_{33}^{(1)}(t)+4\Phi_{j3}^{(2)}(t),
\ee
\be\label{g2}
g_2^T=-2|g_0|^2g_0^T+2ig_1^T\Phi_{33}^{(1)}(t)+4g_0^T\Phi_{33}^{(2)}(t)-8i\Phi_{j3}^{(3)}(t).
\ee
\end{subequations}
We will also need the asymptotic of $c_j(t,k)$,
\begin{lemma}
The global relation (\ref{globalrelj3}) implies that the large $k$ behavior of $c_j(t,k)$ satisfies
\be\label{cjlargek}
c_j(t,k)=\frac{\Phi_{j3}^{(1)}(t)}{k}+\frac{\Phi_{j3}^{(2)}(t)}{k}+\frac{\Phi_{j3}^{(3)}(t)}{k}+O(\frac{1}{k^4}),\quad k\rightarrow \infty,k\in D_1.
\ee
\end{lemma}
\begin{proof}
See the appendix \ref{apdix2}.
\end{proof}

\subsection{The Dirichlet and Neumann problems}
We can now derive effective characterizations of spectral function $S(k)$ for the Dirichlet ($g_0$ prescribed), the first Neumann ($g_1$ prescribed), and the second Neumann ($g_2$ prescribed) problems.
\par
Define $\alpha$ by $\alpha=e^{\frac{2\pi i}{3}}$ and let $\{\Pi_j(t,k),\hat \Pi_j(t,k),\tilde \Pi_j(t,k)\}_1^3$ denote the following combinations formed from $\{\Phi_{j3}(t,k)\}_1^3$:
\be
\ba{l}
\Pi_j(t,k)=\Phi_{j3}(t,k)+\alpha\Phi_{j3}(t,\alpha k)+\alpha^2\Phi_{j3}(t,\alpha^2k),\quad j=1,2,3,\\
\hat \Pi_j(t,k)=\Phi_{j3}(t,k)+\alpha^2\Phi_{j3}(t,\alpha k)+\alpha\Phi_{j3}(t,\alpha^2k),\quad j=1,2,3,\\
\tilde \Pi_j(t,k)=\Phi_{j3}(t,k)+\Phi_{j3}(t,\alpha k)+\Phi_{j3}(t,\alpha^2k),\quad j=1,2,3.
\ea
\ee
And let $R(k)=\Phi_{11}\frac{s_{13}}{s_{33}}+\Phi_{12}\frac{s_{23}}{s_{33}}$.
\par
Let $D_1=D_1^{'}\cup D_1^{''}$ where $D_1^{'}=D_1\cap \{\re k>0\}$ and $D_1^{''}=D_1\cap \{\re k<0\}$. Similarly, let $D_4=D_4^{'}\cup D_4^{''}$ where $D_4^{'}=D_4\cap \{\re k>0\}$ and $D_4^{''}=D_1\cap \{\re k<0\}$.
\begin{theorem}
Let $T<\infty$. Let $u_0(x),u\ge 0$, be a function of Schwartz class.
\par
For the Dirichlet problem it is assumed that the function $g_0(t),0\le t<T$, has sufficient smoothness and is compatible with $u_0(x)$ at $x=t=0$.
\par
For the first Neumann problem it is assumed that the function $g_1(t),0\le t<T$, has sufficient smoothness and is compatible with $u_0(x)$ at $x=t=0$.
\par
Similarly, for the second Neumann problem it is assumed that the function $g_2(t),0\le t<T$, has sufficient smoothness and is compatible with $u_0(x)$ at $x=t=0$.
\par
Suppose that $s_{33}(k)$ has a finite number of simple zeros in $D_1$.
\par
Then the spectral function $S(k)$ is given by
\be\label{Sk}
S(k)=\left(\ba{ccc}A(k)&B(k)&e^{8ik^3T}C(k)\\D(k)&E(k)&e^{8ik^3T}F(k)\\e^{-8ik^3T}G(k)&e^{-8ik^3T}H(k)&I(k)\ea\right)
\ee
where
\[
\ba{ll}
A(k)=\Phi_{22}(k)\Phi_{33}(k)-\Phi_{23}(k)\Phi_{32}(k)&B(k)=\Phi_{13}(k)\Phi_{22}(k)-\Phi_{12}(k)\Phi_{33}(k)\\
C(k)=\Phi_{12}(k)\Phi_{23}(k)-\Phi_{13}(k)\Phi_{22}(k)&D(k)=\Phi_{23}(k)\Phi_{31}(k)-\Phi_{21}(k)\Phi_{33}(k)\\
E(k)=\Phi_{11}(k)\Phi_{33}(k)-\Phi_{13}(k)\Phi_{31}(k)&F(k)=\Phi_{21}(k)\Phi_{13}(k)-\Phi_{11}(k)\Phi_{23}(k)\\
G(k)=\Phi_{21}(k)\Phi_{32}(k)-\Phi_{22}(k)\Phi_{31}(k)&H(k)=\Phi_{12}(k)\Phi_{31}(k)-\Phi_{11}(k)\Phi_{32}(k)\\
I(k)=\Phi_{11}(k)\Phi_{22}(k)-\Phi_{12}(k)\Phi_{21}(k)&
\ea
\]
and the complex-value functions $\{\Phi_{l3}(t,k)\}_{l=1}^{3}$ satisfy the following system of integral equations:
\begin{subequations}\label{Phil3sys}
\be\label{Phi13sys}
\ba{rl}
\Phi_{13}(t,k)&=\int_0^te^{-8ik^3(t-t')}\left[(2ik|g_0|^2+(g_0\bar g_1-g_1\bar g_0))\Phi_{13}\right.\\
&\left.+g_0^2\Phi_{23}+(4k^2g_0+2ikg_1-4|g_0|^2g_0-g_2)\Phi_{33}\right](t',k)dt'
\ea
\ee
\be\label{Phi23sys}
\ba{rl}
\Phi_{23}(t,k)&=\int_0^te^{-8ik^3(t-t')}\left[(2ik|g_0|^2-(g_0\bar g_1-g_1\bar g_0))\Phi_{13}\right.\\
&\left.+\bar g_0^2\Phi_{23}+(4k^2\bar g_0+2ik\bar g_1-4|g_0|^2\bar g_0-\bar g_2)\Phi_{33}\right](t',k)dt'
\ea
\ee
\be\label{Phi33sys}
\ba{rl}
\Phi_{33}(t,k)&=1+\int_0^t\left[(-4k^2\bar g_0+2ik\bar g_1+4|g_0|^2+\bar g_2)\Phi_{13}\right.\\&\left.+(-4k^2g_0+2ikg_1+4|g_0|^2+g_2)\Phi_{23}+-4ik|g_0|^2\Phi_{33}\right](t',k)dt'
\ea
\ee
\end{subequations}
and $\{\Phi_{l1}(t,k)\}_{l=1}^{3},\{\Phi_{l2}(t,k)\}_{l=1}^{3}$ satisfy the following system of integral equations:
\begin{subequations}\label{Phil1sys}
\be\label{Phi11sys}
\ba{rl}
\Phi_{11}(t,k)&=1+\int_0^t\left[(2ik|g_0|^2+(g_0\bar g_1-g_1\bar g_0))\Phi_{11}\right.\\
&\left.+g_0^2\Phi_{21}+(4k^2g_0+2ikg_1-4|g_0|^2g_0-g_2)\Phi_{31}\right](t',k)dt'
\ea
\ee
\be\label{Phi21sys}
\ba{rl}
\Phi_{21}(t,k)&=\int_0^t\left[(2ik|g_0|^2-(g_0\bar g_1-g_1\bar g_0))\Phi_{11}\right.\\
&\left.+\bar g_0^2\Phi_{21}+(4k^2\bar g_0+2ik\bar g_1-4|g_0|^2\bar g_0-\bar g_2)\Phi_{31}\right](t',k)dt'
\ea
\ee
\be\label{Phi33sys}
\ba{rl}
\Phi_{33}(t,k)&=\int_0^te^{8ik^3(t-t')}\left[(-4k^2\bar g_0+2ik\bar g_1+4|g_0|^2+\bar g_2)\Phi_{11}\right.\\&\left.+(-4k^2g_0+2ikg_1+4|g_0|^2+g_2)\Phi_{21}+-4ik|g_0|^2\Phi_{31}\right](t',k)dt'
\ea
\ee
\end{subequations}
\begin{subequations}\label{Phil2sys}
\be\label{Phi12sys}
\ba{rl}
\Phi_{12}(t,k)&=\int_0^t\left[(2ik|g_0|^2+(g_0\bar g_1-g_1\bar g_0))\Phi_{12}\right.\\
&\left.+g_0^2\Phi_{22}+(4k^2g_0+2ikg_1-4|g_0|^2g_0-g_2)\Phi_{32}\right](t',k)dt'
\ea
\ee
\be\label{Phi22sys}
\ba{rl}
\Phi_{22}(t,k)&=1+\int_0^t\left[(2ik|g_0|^2-(g_0\bar g_1-g_1\bar g_0))\Phi_{12}\right.\\
&\left.+\bar g_0^2\Phi_{22}+(4k^2\bar g_0+2ik\bar g_1-4|g_0|^2\bar g_0-\bar g_2)\Phi_{32}\right](t',k)dt'
\ea
\ee
\be\label{Phi32sys}
\ba{rl}
\Phi_{32}(t,k)&=\int_0^te^{8ik^3(t-t')}\left[(-4k^2\bar g_0+2ik\bar g_1+4|g_0|^2+\bar g_2)\Phi_{12}\right.\\&\left.+(-4k^2g_0+2ikg_1+4|g_0|^2+g_2)\Phi_{22}+-4ik|g_0|^2\Phi_{32}\right](t',k)dt'
\ea
\ee
\end{subequations}
\begin{enumerate}
\item For the Dirichlet problem, the unknown Neumann boundary values $g_1(t)$ and $g_2(t)$ are given by
\begin{subequations}\label{DtoN}
\be\label{DtoNg1}
\ba{rl}
g_1(t)=&\frac{2g_0(t)}{\pi}\int_{\pt D_3}\Pi_3(t,k)dk+\frac{2}{\pi i}\int_{\pt D_3}\left[k\Pi_1(t,k)-\frac{3g_0(t)}{2i}\right]dk\\
&-\frac{2}{\pi i}\int_{\pt D_3}ke^{-8ik^3t}[(\alpha^2-\alpha)R(\alpha k)+(\alpha-\alpha^2)R(\alpha^2k)]dk\\
&+4\left\{(1-\alpha^2)\sum_{k_j\in D_1^{'}}+(1-\alpha)\sum_{k_j\in D_1^{''}}\right\}k_je^{-8ik_j^3t}Res_{k_j}R(k).
\ea
\ee
and
\be\label{DtoNg2}
\ba{rl}
g_2(t)=&g_0(t)^3-\frac{4}{\pi}\int_{\pt D_3}\left[k^2\Pi_1(t,k)-\frac{3kg_0(t)}{2i}\right]dk\\
&+\frac{4}{\pi}\int_{\pt D_3}k^2e^{-8ik^3t}\left[(1-\alpha)R(\alpha k)+(1-\alpha^2)R(\alpha^2k)\right]dk\\
&-8i\left\{(1-\alpha)\sum_{k_j\in D_1^{'}}+(1-\alpha^2)\sum_{k_j\in D_1^{''}}\right\}k_j^2e^{-8ik_j^3t}Res_{k_j}R(k)\\
&+\frac{4g_0(t)}{\pi i}\int_{\pt D_3}k\hat \Pi_3(t,k)dk+\frac{2g_1(t)}{\pi}\int_{\pt D_3}\Pi_3(t,k)dk.
\ea
\ee
\end{subequations}
\item For the first Neumann problem, the unknown boundary values $g_0(t)$ and $g_2(t)$ are given by
\begin{subequations}\label{FNtoD}
\be\label{FNtoDg0}
\ba{rl}
g_0(t)=&\frac{1}{\pi}\int_{\pt D_3}\hat \Pi_1(t,k)dk-\frac{1}{\pi}\int_{\pt D_3}e^{-8ik^3t}\left[(\alpha-alpha^2)R(\alpha k)+(\alpha^2-\alpha)R(\alpha^2k)\right]dk\\
&+2i\left\{(1-\alpha)\sum_{k_j\in D_1^{'}}+(1-\alpha^2)\sum_{k_j\in D_1^{''}}\right\}e^{-8ik_j^3t}Res_{k_j}R(k),
\ea
\ee
and
\be\label{FNtoDg2}
\ba{rl}
g_2(t)=&g_0^3(t)-\frac{4}{\pi}\int_{\pt D_3}\left(k^2\hat \Pi_1(t,k)-\frac{3}{\pi i}\int_{\pt D_3}l\hat \Pi_1(t,l)dl\right)dk\\
&+\frac{4}{\pi}\int_{\pt D_3}k^2e^{-8ik^3t}\left[(1-\alpha^2)R(\alpha k)+(1-\alpha)R(\alpha^2k)\right]dk\\
&-8i\left\{(1-\alpha)\sum_{k_j\in D_1^{'}}+(1-\alpha^2)\sum_{k_j\in D_1^{''}}\right\}k_j^2e^{-8ik_j^3t}Res_{k_j}R(k)\\
&+\frac{4g_0(t)}{\pi i}\int_{\pt D_3}k\hat \Pi_3(t,k)dk+\frac{2g_1(t)}{\pi}\int_{\pt D_3}\Pi_3(t,k)dk.
\ea
\ee
\end{subequations}
\item For the second Neumann problem, the unknown boundary values $g_0(t)$ and $g_1(t)$ are given by
\begin{subequations}\label{SNtoD}
\be\label{SNtoDg0}
\ba{rl}
g_0(t)=&\frac{1}{\pi}\int_{\pt D_3}\hat \Pi_1(t,k)dk-\frac{1}{\pi}\int_{\pt D_3}e^{-8ik^3t}\left[(\alpha-alpha^2)R(\alpha k)+(\alpha^2-\alpha)R(\alpha^2k)\right]dk\\
&+2i\left\{(1-\alpha)\sum_{k_j\in D_1^{'}}+(1-\alpha^2)\sum_{k_j\in D_1^{''}}\right\}e^{-8ik_j^3t}Res_{k_j}R(k),
\ea
\ee
and
\be\label{SNtoDg1}
\ba{rl}
g_1(t)=&\frac{2g_0(t)}{\pi}\int_{\pt D_3}\Pi_3(t,k)dk+\frac{2}{\pi i}\int_{\pt D_3}k\tilde \Pi_1(t,k)dk\\
&-\frac{2}{\pi i}\int_{\pt D_3}ke^{-8ik^3t}\left[(\alpha^2-1)R(\alpha k)+(\alpha-1)R(\alpha^2k)\right]dk\\
&+4\left\{(1-\alpha)\sum_{k_j\in D_1^{'}}+(1-\alpha^2)\sum_{k_j\in D_1^{''}}\right\}k_je^{-8ik_j^3t}Res_{k_j}R(k).
\ea
\ee
\end{subequations}
\end{enumerate}
\end{theorem}
\begin{proof}
The representations (\ref{Sk}) follow from the relation $S(k)=e^{8ik^3T}\mu_2^A(0,T,k)^T$. And the system (\ref{Phil3sys}) is the direct result of the Volteral integral equations of $\mu_2(0,t,k)$.
\begin{enumerate}
\item In order to derive (\ref{DtoNg1}) we note that equation (\ref{g1}) expresses $g_1$ in terms of $\Phi_{33}^{(1)}$ and $\Phi_{13}^{(2)}$. Furthermore, equation (\ref{Phi3}) and Cauchy theorem imply
    \[
    -\frac{2\pi i}{3}\Phi_{33}^{(1)}(t)=2\int_{\pt D_2}[\Phi_{33}(t,k)-1]dk=\int_{\pt D_4}[\Phi_{33}(t,k)-1]dk
    \]
    and
    \[
    -\frac{2\pi i}{3}\Phi_{13}^{(2)}(t)=2\int_{\pt D_2}\left[k\Phi_{13}(t)-\frac{g_0(t)}{2i}\right]dk=\int_{\pt D_4}\left[k\Phi_{13}(t)-\frac{g_0(t)}{2i}\right]dk.
    \]
    Thus,
    \be\label{Phi331}
    \ba{l}
    i\pi\Phi_{33}^{(1)}(t)=-\left(\int_{\pt D_2}+\int_{\pt D_4}\right)[\Phi_{33}(t,k)-1]dk=\left(\int_{\pt D_1}+\int_{\pt D_3}\right)[\Phi_{33}(t,k)-1]dk\\
    =\int_{\pt D_3}[\Phi_{33}(t,k)-1]dk+\alpha\int_{\pt D_3}[\Phi_{33}(t,k)-1]dk+\alpha^2\int_{\pt D_3}[\Phi_{33}(t,k)-1]dk\\
    =\int_{\pt D_3}\Pi_{3}(t,k)dk.
    \ea
    \ee
    Similarly,
    \be\label{Phi132}
    \ba{l}
    i\pi \Phi_{13}^{(2)}(t)=\left(\int_{\pt D_3}+\int_{\pt D_1}\right)\left[k\Phi_{13}(t)-\frac{g_0(t)}{2i}\right]dk\\
    =\left(\int_{\pt D_3}+\alpha^2\int_{\pt D_1^{'}}+\alpha\int_{\pt D_1^{''}}\right)\left[k\Phi_{13}(t)-\frac{g_0(t)}{2i}\right]dk+I(t)\\
    =\int_{\pt D_3}\left[k\Pi_{1}(t,k)-\frac{3g_0(t)}{2i}\right]dk+I(t).
    \ea
    \ee
    where $I(t)$ is defined by
    \[
    I(t)=\left((1-\alpha^2)\int_{\pt D_1^{'}}+(1-\alpha)\int_{\pt D_1^{''}}\right)\left[k\Phi_{13}(t)-\frac{g_0(t)}{2i}\right]dk
    \]
    The last step involves using the global relation to compute $I(t)$
    \be\label{DtoNIt}
    \ba{r}
    I(t)=\left((1-\alpha^2)\int_{\pt D_1^{'}}+(1-\alpha)\int_{\pt D_1^{''}}\right)\left[kc_1(t,k)-\frac{g_0(t)}{2i}\right]dk\\
    -\left((1-\alpha^2)\int_{\pt D_1^{'}}+(1-\alpha)\int_{\pt D_1^{''}}\right)ke^{-8ik^3t}R(k)dk
    \ea
    \ee
    Using the asymptotic (\ref{cjlargek}) and Cauchy theorem to compute the first term on the right-hand side of equation (\ref{DtoNIt}) and using the transformation $k\rightarrow\alpha k$ and $k\rightarrow \alpha^2k$ in the second term on the right-hand side of (\ref{DtoNIt}), we find
    \be\label{DtoNItres}
    \ba{r}
    I(t)=-i\pi \Phi_{13}^{(2)}(t)-\int_{\pt D_3}ke^{-8ik^3t}\left[(\alpha^2-\alpha)R(\alpha k)+(\alpha-\alpha^2)R(\alpha^2k)\right]dk\\
    +2\pi i\left\{(1-\alpha^2)\sum_{k_j\in D_1^{'}}+(1-\alpha)\sum_{k_j\in D_1^{''}}\right\}ke^{-8ik_j^3t}Res_{k_j}R(k).
    \ea
    \ee
    Equations (\ref{Phi132}) and (\ref{DtoNItres}) imply
    \[
    \ba{l}
    \Phi_{13}^{(2)}(t)=\frac{1}{2\pi i}\int_{\pt D_3}\left[k\Pi_{1}(t,k)-\frac{3g_0(t)}{2i}\right]dk\\
    -\frac{1}{2\pi i}\int_{\pt D_3}ke^{-8ik^3t}\left[(\alpha^2-\alpha)R(\alpha k)+(\alpha-\alpha^2)R(\alpha^2k)\right]dk\\
    \left\{(1-\alpha^2)\sum_{k_j\in D_1^{'}}+(1-\alpha)\sum_{k_j\in D_1^{''}}\right\}k_je^{-8ik_j^3t}Res_{k_j}R(k).
    \ea
    \]
    This equation together with (\ref{g1}) and (\ref{Phi331}) yields (\ref{DtoNg1}).
    \par
    In order to derive (\ref{DtoNg2}), we note that (\ref{g2}) expresses $g_2$ in terms of $\Phi_{13}^{(3)}$, $\Phi_{33}^{(2)}$ and $\Phi_{33}^{(1)}$. Equation (\ref{DtoNg2}) follows from the expression (\ref{Phi331}) for $\Phi_{33}^{(1)}$ and the following formulas:
    \begin{subequations}
    \be\label{Phi332}
    \Phi_{33}^{(2)}(t)=\frac{1}{\pi i}\int_{\pt D_3}k\hat \Pi_3dk,
    \ee
    \be\label{Phi133}
    \ba{l}
    \Phi_{13}^{(3)}(t)=\frac{1}{2\pi i}\int_{\pt D_3}\left[k^2\Pi_{1}(t,k)-\frac{3kg_0(t)}{2i}\right]dk\\
    -\frac{1}{2\pi i}\int_{\pt D_3}k^2e^{-8ik^3t}\left[(1-\alpha)R(\alpha k)+(1-\alpha^2)R(\alpha^2k)\right]dk\\
    \left\{(1-\alpha)\sum_{k_j\in D_1^{'}}+(1-\alpha^2)\sum_{k_j\in D_1^{''}}\right\}k_j^2e^{-8ik_j^3t}Res_{k_j}R(k).
    \ea
    \ee
    \end{subequations}
    \item In order to derive the representations (\ref{FNtoD}) relevant for the first Neumann problem, we use (\ref{g012}) together with (\ref{Phi331}), (\ref{Phi332}) and the following formulas:
        \begin{subequations}
        \be\label{Phi131}
        \ba{l}
        \Phi_{13}^{(1)}(t)=\frac{1}{2\pi i}\int_{\pt D_3}\hat \Pi_{1}(t,k)dk\\
         -\frac{1}{2\pi i}\int_{\pt D_3}e^{-8ik^3t}\left[(\alpha-\alpha^2)R(\alpha k)+(\alpha^2-\alpha)R(\alpha^2k)\right]dk\\
        \left\{(1-\alpha)\sum_{k_j\in D_1^{'}}+(1-\alpha^2)\sum_{k_j\in D_1^{''}}\right\}e^{-8ik_j^3t}Res_{k_j}R(k).
        \ea
        \ee
        \be\label{Phi132FN}
        \Phi_{13}^{(2)}(t)=\frac{1}{\pi i}\int_{\pt D_3}k\hat \Pi_1dk,
        \ee
        \be\label{Phi133FN}
        \ba{l}
        \Phi_{13}^{(3)}(t)=\frac{1}{2\pi i}\int_{\pt D_3}\left[k^2\hat \Pi_{1}(t,k)-3\Phi_{13}^{(2)}\right]dk\\
         -\frac{1}{2\pi i}\int_{\pt D_3}k^2e^{-8ik^3t}\left[(1-\alpha^2)R(\alpha k)+(1-\alpha)R(\alpha^2k)\right]dk\\
        \left\{(1-\alpha^2)\sum_{k_j\in D_1^{'}}+(1-\alpha)\sum_{k_j\in D_1^{''}}\right\}k_j^2e^{-8ik_j^3t}Res_{k_j}R(k).
        \ea
        \ee
        \end{subequations}
    \item In order to derive the representations (\ref{SNtoD}) relevant for the second Neumann problem, we use (\ref{g012}) together with (\ref{Phi331}) and the following formulas:
        \begin{subequations}
        \be\label{Phi131SN}
        \ba{l}
        \Phi_{13}^{(1)}(t)=\frac{1}{2\pi i}\int_{\pt D_3}\tilde \Pi_{1}(t,k)dk\\
         -\frac{1}{2\pi i}\int_{\pt D_3}e^{-8ik^3t}\left[(\alpha-1)R(\alpha k)+(\alpha^2-1)R(\alpha^2k)\right]dk\\
        \left\{(1-\alpha^2)\sum_{k_j\in D_1^{'}}+(1-\alpha)\sum_{k_j\in D_1^{''}}\right\}e^{-8ik_j^3t}Res_{k_j}R(k).
        \ea
        \ee
        \be\label{Phi132SN}
        \ba{l}
        \Phi_{13}^{(2)}(t)=\frac{1}{2\pi i}\int_{\pt D_3}k\tilde \Pi_{1}(t,k)dk\\
         -\frac{1}{2\pi i}\int_{\pt D_3}ke^{-8ik^3t}\left[(\alpha^2-1)R(\alpha k)+(\alpha-1)R(\alpha^2k)\right]dk\\
        \left\{(1-\alpha)\sum_{k_j\in D_1^{'}}+(1-\alpha^2)\sum_{k_j\in D_1^{''}}\right\}e^{-8ik_j^3t}Res_{k_j}R(k).
        \ea
        \ee
        \end{subequations}
\end{enumerate}
\end{proof}

\subsection{Effective characterizations}
Substituting into the system (\ref{Phil3sys}) the expressions
\begin{subequations}
\be\label{Phij3eps}
\Phi_{ij}=\Phi^{(0)}_{ij}+\eps\Phi^{(1)}_{ij}+\eps^2\Phi^{(2)}_{ij}+\cdots,\quad i,j=1,2,3.
\ee
\be\label{g0eps}
g_0=\eps g_{01}+\eps^2 g_{02}+\cdots,
\ee
\be\label{g1eps}
g_1=\eps g_{11}+\eps^2 g_{12}+\cdots,
\ee
\be\label{g2eps}
g_2=\eps g_{21}+\eps^2 g_{22}+\cdots,
\ee
\end{subequations}
where $\eps>0$ is a small parameter, we find that the terms of $O(1)$ give $\Phi^{(0)}_{13}=\Phi^{(0)}_{23}=0$ and $\Phi^{(0)}_{33}=1$. Moreover, the terms of $O(\eps)$ give $\Phi^{(1)}_{33}=0$ and
\be\label{Oeps}
O(\eps):\quad \Phi^{(1)}_{13}(t,k)=\int_0^te^{-8ik^3(t-t')}(4k^2g_{01}+2ikg_{11}-g_{21})(t',k)dt',
\ee
From the above equation (\ref{Oeps}) we can get
\begin{subequations}
\be\label{Pi11}
\Pi^{(1)}_{1}(t,k)=12k^2\int_0^te^{-8ik^3(t-t')}g_{01}(t')dt',
\ee
\be\label{hatPi11}
\hat \Pi^{(1)}_{1}(t,k)=6ik\int_0^te^{-8ik^3(t-t')}g_{11}(t')dt',
\ee
\be\label{tildePi11}
\tilde \Pi^{(1)}_{1}(t,k)=-3\int_0^te^{-8ik^3(t-t')}g_{11}(t')dt',
\ee
\end{subequations}
\par
The Dirichlet problem can now be solved perturbatively as follows: assuming for simplicity that $s_{33}(k)$ has no zeros and expanding (\ref{DtoNg1}) and (\ref{DtoNg2}), we find
\begin{subequations}\label{DtoN1}
\be\label{DtoNg11}
\ba{rl}
g_{11}=&\frac{2}{\pi i}\int_{\pt D_3}\left[k\Pi^{(1)}_1(t,k)-\frac{3g_{01}(t)}{2i}\right]dk\\
&-\frac{2}{\pi i}\int_{\pt D_3}ke^{-8ik^3t}[(\alpha^2-\alpha)s_{131}(\alpha k)+(\alpha-\alpha^2)s_{131}(\alpha^2k)]dk
\ea
\ee
\be\label{DtoNg21}
\ba{rl}
g_{21}=&-\frac{4}{\pi}\int_{\pt D_3}\left[k^2\Pi^{(1)}_1(t,k)-\frac{3kg_{01}(t)}{2i}\right]dk\\
&+\frac{4}{\pi}\int_{\pt D_3}k^2e^{-8ik^3t}\left[(1-\alpha)s_{131}(\alpha k)+(1-\alpha^2)s_{131}(\alpha^2k)\right]dk
\ea
\ee
\end{subequations}
Using equation (\ref{Pi11}) to determine $\Pi^{(1)}_{1}$, we can determine $g_{11},g_{21}$ from (\ref{DtoN1}), then $\Phi^{(1)}_{13}$ can be found from (\ref{Oeps}), And these arguments can be extended to higher orders and also can be extended to the systems (\ref{Phi11sys}) and (\ref{Phi12sys}), thus yields a constructive scheme for computing $S(k)$ to all orders.
\par
Similarly, these arguments also can be used to the first Neumann problem and the second Neumann problem. That is to say, in all cases, the system can be solved perturbatively to all orders.

\appendix
\section{The asymptotic behavior of the functions $\{\mu_j(x,t,k)\}_1^3$}\label{apdix1}
We denote some symbols as follows:
\begin{subequations}
\be\label{MLam}
\Lam=\left(\ba{ll}\id_{2\times 2}&0\\0&-1\ea\right),
\ee
\be\label{MV}
\ba{l}
V_1=\left(\ba{ll}0&U^T\\-\bar U&0\ea\right),\\
V_2^{(2)}=4\left(\ba{ll}0&U^T\\-\bar U&0\ea\right),\\
V_2^{(1)}=2i\left(\ba{ll}U^T\bar U&U_x^T\\\bar U_x&-2|u|^2\ea\right),\\
V_2^{(0)}=-4|u|^2\left(\ba{ll}0&U^T\\-\bar U&0\ea\right)-\left(\ba{ll}0&U_{xx}^T\\-\bar U_{xx}&0\ea\right)+(u\bar u_x-u_x\bar u)\left(\ba{ll}\sig_3&0\\0&0\ea\right).
\ea
\ee
where $\id_{2\times 2}=\left(\ba{ll}1&0\\0&1\ea\right)$ and $U=(u,\bar u)$.
\end{subequations}
\par
From the Lax pair of $\mu$
\be\label{muLaxe}
\left\{
\ba{l}
\mu_x+[ik\Lam,\mu]=V_1\mu,\\
\mu_t+[4ik^3\Lam,\mu]=V_2\mu.
\ea
\right.
\ee
Suppose that
\be\label{muklarg}
\mu(x,t,k)=D_0+\frac{D_1}{k}+\frac{D_2}{k^2}+\frac{D_3}{k^3}+\cdots.
\ee
We substitute the equation (\ref{muklarg}) into the Lax pair (\ref{muLaxe}), and compare the order of $k$, we find that:
\begin{subequations}
\be\label{mulargkx}
\ba{ll}
O(k):&[i\Lam,D_0]=0,\\
O(1):&D_{0x}+[i\Lam,D_1]=V_1D_0,\\
O(k^{-1}):&D_{1x}+[i\Lam,D_2]=V_1D_1,\\
O(k^{-2}):&D_{2x}+[i\Lam,D_3]=V_1D_2,\\
\ea
\ee
\be\label{mulargkt}
\ba{ll}
O(k^3):&[4i\Lam,D_0]=0,\\
O(k^2):&[4i\Lam,D_1]=V_2^{(2)}D_0,\\
O(k^1):&[4i\Lam,D_2]=V_2^{(2)}D_1+V_2^{(1)}D_0,\\
O(1):&D_{0t}+[4i\Lam,D_3]=V_2^{(2)}D_2+V_2^{(1)}D_1+V_2^{(0)}D_0,\\
O(k^{-1}):&D_{1t}+[4i\Lam,D_4]=V_2^{(2)}D_3+V_2^{(1)}D_4+V_2^{(0)}D_1,\\
O(k^{-2}):&D_{2t}+[4i\Lam,D_5]=V_2^{(2)}D_4+V_2^{(1)}D_3+V_2^{(0)}D_2,\\
\ea
\ee
\end{subequations}
And we denote the $D_l$ by $D_l=\left(\ba{ll}D_{2\times 2}^{(l)}&D_{j3}^{(l)}\\D_{3j}^{(l)}&D_{33}^{(l)}\ea\right),\quad j=1,2$.
\par
Then, from $O(k^3)$,we have
\be\label{Ok3}
D_{j3}^{(0)}=0,\quad D_{3j}^{(0)}=0.
\ee
$O(k^2)$, we get
\begin{subequations}
\be\label{Ok2}
4i\left(\ba{ll}0&2D_{j3}^{(1)}\\-2D_{3j}^{(1)}&0\ea\right)=4\left(\ba{ll}0&U^TD_{33}^{(0)}\\-\bar U D_{2\times 2}^{(0)}&0\ea\right),
\ee
this implies that
\be\label{Ok2rsu}
\left\{
\ba{l}
D_{j3}^{(1)}=-\frac{i}{2}U^TD_{33}^{(0)}\\
D_{3j}^{(1)}=-\frac{i}{2}\bar UD_{2\times 2}^{(0)}.
\ea
\right.
\ee
\end{subequations}
$O(k)$, we find
\begin{subequations}
\be\label{Ok1}
\ba{l}
4i\left(\ba{ll}0&2D_{j3}^{(2)}\\-2D_{3j}^{(2)}&0\ea\right)=\\
4\left(\ba{ll}U^TD_{3j}^{(1)}&U^TD_{33}^{(1)}\\-\bar UD_{2\times 2}^{(1)}&-\bar UD_{j3}^{(1)}\ea\right)+2i\left(\ba{ll}U^T\bar UD_{2\times 2}^{(0)}&U_x^TD_{33}^{(0)}\\-\bar U_xD_{2\times 2}^{(0)}&-2|u|^2D_{33}^{(0)}\ea\right),
\ea
\ee
this implies that
\be\label{Ok1rsu}
\left\{
\ba{l}
D_{j3}^{(2)}=-\frac{i}{2}U^TD_{33}^{(1)}+\frac{1}{4}U_x^TD_{33}^{0}\\
D_{3j}^{(2)}=-\frac{i}{2}\bar UD_{2\times 2}^{(1)}-\frac{1}{4}\bar U_xD_{2\times 2}^{(0)}.
\ea
\right.
\ee
\end{subequations}
$O(1)$, we have
\begin{subequations}
\be\label{Ok0}
\ba{l}
\left(\ba{ll}D_{2\times 2t}^{(0)}&0\\0&D_{33t}^{(0)}\ea\right)+4i\left(\ba{ll}0&2D_{j3}^{(3)}\\-2D_{3j}^{(3)}&0\ea\right)=\\
4\left(\ba{ll}U^TD_{3j}^{(2)}&U^TD_{33}^{(2)}\\-\bar UD_{2\times 2}^{(2)}&-\bar UD_{j3}^{(2)}\ea\right)+2i\left(\ba{ll}U^T\bar UD_{2\times 2}^{(1)}+U_x^TD_{3j}^{(1)}&U^T\bar UD_{j3}^{(1)}+U_x^TD_{33}^{(1)}\\-\bar U_xD_{2\times 2}^{(1)}-2|u|^2D_{3j}^{(1)}&\bar U_xD_{j3}^{(1)}-2|u|^2D_{33}^{(1)}\ea\right)\\
-4|u|^2\left(\ba{ll}0&U^TD_{33}^{(0)}\\-\bar U D_{2\times 2}^{(0)}&0\ea\right)-\left(\ba{ll}0&U_{xx}^TD_{33}^{(0)}\\-\bar U_{xx} D_{2\times 2}^{(0)}&0\ea\right)\\
+(u\bar u_x-u_x\bar u)\left(\ba{ll}\sig_3D_{2\times 2}^{(0)}&0\\0&0\ea\right).
\ea
\ee
this implies that
\be\label{Ok0rsu}
\ba{l}
D_{2\times 2t}^{(0)}=0\quad D_{33t}^{(0)}=0\\
\left\{
\ba{l}
D_{j3}^{(3)}=-\frac{i}{2}U^TD_{33}^{(2)}+\frac{1}{4}U_x^TD_{33}^{(1)}+\frac{i}{4}|u|^2U^TD_{33}^{(0)}+\frac{i}{8}U_{xx}^TD_{33}^{(0)}\\
D_{3j}^{(3)}=-\frac{i}{2}\bar UD_{2\times 2}^{(2)}-\frac{1}{4}\bar U_xD_{2\times 2}^{(1)}+\frac{i}{4}|u|^2\bar UD_{2\times 2}^{(0)}+\frac{i}{8}\bar U_{xx}D_{2\times 2}^{(0)}.
\ea
\right.
\ea
\ee
\end{subequations}
$O(k^{-1})$, we get
\begin{subequations}
\be\label{Ok-1}
\ba{l}
\left(\ba{ll}D_{2\times 2t}^{(1)}&D_{j3t}^{(1)}\\D_{3jt}^{(1)}&D_{33t}^{(1)}\ea\right)+4i\left(\ba{ll}0&2D_{j3}^{(4)}\\-2D_{3j}^{(4)}&0\ea\right)=\\
4\left(\ba{ll}U^TD_{3j}^{(3)}&U^TD_{33}^{(3)}\\-\bar UD_{2\times 2}^{(3)}&-\bar UD_{j3}^{(3)}\ea\right)+2i\left(\ba{ll}U^T\bar UD_{2\times 2}^{(2)}+U_x^TD_{3j}^{(2)}&U^T\bar UD_{j3}^{(2)}+U_x^TD_{33}^{(2)}\\-\bar U_xD_{2\times 2}^{(2)}-2|u|^2D_{3j}^{(2)}&\bar U_xD_{j3}^{(2)}-2|u|^2D_{33}^{(2)}\ea\right)\\
-4|u|^2\left(\ba{ll}U^TD_{3j}^{(1)}&U^TD_{33}^{(1)}\\-\bar U D_{2\times 2}^{(1)}&-\bar UD_{j3}^{(1)}\ea\right)-\left(\ba{ll}U_{xx}^TD_{3j}^{(1)}&U_{xx}^TD_{33}^{(1)}\\-\bar U_{xx} D_{2\times 2}^{(1)}&-\bar U_{xx} D_{j3}^{(1)}\ea\right)\\
+(u\bar u_x-u_x\bar u)\left(\ba{ll}\sig_3D_{2\times 2}^{(1)}&\sig_3D_{j3}^{(1)}\\0&0\ea\right).
\ea
\ee
this implies that
\be\label{Ok-1rsu}
\ba{l}
\left\{
\ba{l}
D_{2\times 2t}^{(1)}=\frac{i}{2}\{U^T\bar U_{xx}+U_{xx}\bar U-U_x^T\bar U_x+6|u|^2U^T\bar U\}D_{2\times 2}^{(0)}\\
D_{33t}^{(1)}=-i\{u\bar u_{xx}+u_{xx}\bar u-u_x\bar u_x+6|u|^4\}D_{33}^{(0)}.
\ea
\right.\\
\left\{
\ba{l}
D_{j3}^{(4)}=\frac{1}{16}U_t^TD_{33}^{(0)}-\frac{i}{2}U^TD_{33}^{(3)}+\frac{1}{4}U_x^TD_{33}^{(2)}+\frac{i}{4}|u|^2U^TD_{33}^{(1)}+\frac{i}{8}U_{xx}^TD_{33}^{(1)}+\frac{1}{8}|u|^2U_x^TD_{33}^{(0)}\\
D_{3j}^{(3)}=-\frac{1}{16}\bar U_tD_{2\times 2}^{(0)}-\frac{i}{2}\bar UD_{2\times 2}^{(3)}-\frac{1}{4}\bar U_xD_{2\times 2}^{(2)}+\frac{i}{4}|u|^2\bar UD_{2\times 2}^{(1)}+\frac{i}{8}\bar U_{xx}D_{2\times 2}^{(1)}-\frac{1}{8}|u|^2\bar U_xD_{2\times 2}^{(0)}.
\ea
\right.
\ea
\ee
\end{subequations}
$O(k^{-2})$, we get
\begin{subequations}
\be\label{Ok-2}
\ba{l}
\left(\ba{ll}D_{2\times 2t}^{(2)}&D_{j3t}^{(2)}\\D_{3jt}^{(2)}&D_{33t}^{(2)}\ea\right)+4i\left(\ba{ll}0&2D_{j3}^{(5)}\\-2D_{3j}^{(5)}&0\ea\right)=\\
4\left(\ba{ll}U^TD_{3j}^{(4)}&U^TD_{33}^{(4)}\\-\bar UD_{2\times 2}^{(4)}&-\bar UD_{j3}^{(4)}\ea\right)+2i\left(\ba{ll}U^T\bar UD_{2\times 2}^{(3)}+U_x^TD_{3j}^{(3)}&U^T\bar UD_{j3}^{(3)}+U_x^TD_{33}^{(3)}\\-\bar U_xD_{2\times 2}^{(3)}-2|u|^2D_{3j}^{(3)}&\bar U_xD_{j3}^{(3)}-2|u|^2D_{33}^{(3)}\ea\right)\\
-4|u|^2\left(\ba{ll}U^TD_{3j}^{(2)}&U^TD_{33}^{(2)}\\-\bar U D_{2\times 2}^{(2)}&-\bar UD_{j3}^{(2)}\ea\right)-\left(\ba{ll}U_{xx}^TD_{3j}^{(2)}&U_{xx}^TD_{33}^{(2)}\\-\bar U_{xx} D_{2\times 2}^{(2)}&-\bar U_{xx} D_{j3}^{(2)}\ea\right)\\
+(u\bar u_x-u_x\bar u)\left(\ba{ll}\sig_3D_{2\times 2}^{(2)}&\sig_3D_{j3}^{(2)}\\0&0\ea\right).
\ea
\ee
this implies that
\be\label{Ok-2rsu}
\left\{
\ba{l}
\ba{rl}D_{2\times 2t}^{(2)}=&\frac{i}{2}\{U^T\bar U_{xx}+U_{xx}\bar U-U_x^T\bar U_x+6|u|^2U^T\bar U\}D_{2\times 2}^{(1)}\\&+\{-\frac{1}{4}U^T\bar U_t+\frac{1}{2}|u|^2(u\bar u_x-u_x\bar u)\sig_3+\frac{1}{4}(u_{xx}\bar u_x-u_x\bar u_{xx})\sig_3\}\ea\\
D_{33t}^{(2)}=-i\{u\bar u_{xx}+u_{xx}\bar u-u_x\bar u_x+6|u|^4\}D_{33}^{(1)}-\frac{1}{4}(|u|^2)_tD_{33}^{(0)}.
\ea
\right.
\ee
\end{subequations}
Also, from the $x-$part of the Lax pair, we have the following equations
\begin{subequations}
\be\label{Ok0x}
D_{2\times 2x}^{(0)}=0,\quad D_{33x}^{(0)}=0.
\ee
\be\label{Ok-1x}
\left\{
\ba{l}
D_{2\times 2x}^{(1)}=-\frac{i}{2}U^T\bar UD_{2\times 2}^{(0)}\\
D_{33x}^{(1)}=i|u|^2D_{33}^{(0)}.
\ea
\right.
\ee
\be\label{Ok-2x}
\left\{
\ba{l}
D_{2\times 2x}^{(2)}=-\frac{i}{2}U^T\bar UD_{2\times 2}^{(1)}-\frac{1}{4}U^T\bar U_xD_{2\times 2}^{(0)}\\
D_{33x}^{(2)}=i|u|^2D_{33}^{(1)}-\frac{1}{4}(|u|^2)_xD_{33}^{(0)}.
\ea
\right.
\ee
\end{subequations}
Then from the integral contours $\gam_j$, we can get
\be\label{D0}
D_{2\times 2}^{(0)}=\id_{2\times 2},\quad D_{33}^{(0)}=1.
\ee

\section{The asymptotic behavior of $c_j(t,k)$}\label{apdix2}
Let
\[
\mu_2(0,t,k)=\left(\ba{ll}\Phi_{2\times 2}&\Phi_{j3}\\\Phi_{3j}&\Phi_{33}\ea\right).
\]
The global relation shows that
\be\label{cjdef}
\Phi_{2\times 2}\frac{s_{j3}}{s_{33}}e^{-8ik^3t}+\Phi_{j3}=c_j.
\ee
And from equation
\[
\mu_t+[4ik^3\Lam,\mu]=V_2\mu.
\]
we get
\be\label{Phit}
\ba{l}
\left(\ba{ll}\Phi_{2\times 2}&\Phi_{j3}\\\Phi_{3j}&\Phi_{33}\ea\right)_t+4ik^3\left(\ba{ll}0&2\Phi_{j3}\\-2\Phi_{3j}&0\ea\right)=4k^2\left(\ba{ll}U^T\Phi_{3j}&U^T\Phi_{33}\\-\bar U\Phi_{2\times 2}&-\bar U\Phi_{j3}\ea\right)\\
+2ik\left(\ba{ll}U^T\bar U\Phi_{2\times 2}+U_x^T\Phi_{3j}&U^T\bar U\Phi_{j3}+U_x^T\Phi_{33}\\\bar U_x\Phi_{2\times 2}-2|u|^2\Phi_{3j}&-\bar U_x\Phi_{j3}-2|u|^2\Phi_{33}\ea\right)-4|u|^2\left(\ba{ll}U^T\Phi_{3j}&U^T\Phi_{33}\\-\bar U\Phi_{2\times 2}&-\bar U\Phi_{j3}\ea\right)\\
-\left(\ba{ll}U_{xx}^T\Phi_{3j}&U_{xx}^T\Phi_{33}\\-\bar U_{xx}\Phi_{2\times 2}&-\bar U_{xx}\Phi_{j3}\ea\right)+(u\bar u_x-u_x\bar u)\left(\ba{ll}\sig_3\Phi_{2\times 2}&\sig_3\Phi_{j3}\\0&0\ea\right).
\ea
\ee
From the second column of the equation (\ref{Phit}) we get
\be\label{Phi2t}
\left\{
\ba{l}
\ba{rl}\Phi_{j3t}+8ik^3\Phi_{j3}=&4k^2U^T\Phi_{33}+2ik(U^T\bar U\Phi_{j3}+U^T_x\Phi_{33})\\&-4|u|^2U^T\Phi_{33}-U_{xx}^T\Phi_{33}+(u\bar u_x-u_x\bar u)\sig_3\Phi_{j3}\ea\\
\Phi_{33t}=-4k^2\bar U\Phi_{j3}+2ik(\bar U_x\Phi_{j3}-2|u|^2\Phi_{33})+4|u|^2\bar U\Phi_{j3}+\bar U_{xx}\Phi_{j3}.
\ea
\right.
\ee
Suppose
\be\label{Phi2albt}
\left(\ba{l}\Phi_{j3}\\\Phi_{33}\ea\right)=(\alpha_0(t)+\frac{\alpha_1(t)}{k}+\frac{\alpha_2(t)}{k^2}+\cdots)+(\beta_0(t)+\frac{\beta_1(t)}{k}+\frac{\beta_2(t)}{k^2}+\cdots)e^{-8ik^3t}
\ee
where the coefficients $\alpha_l(t)$ and $\beta_l(t)$, $l\ge 0$, are independent of $k$. To determine these coefficients,we substitute the above equation into equation (\ref{Phi2t}) and use the initial conditions
\[
\alpha_0(0)+\beta_0(0)=(0_{1\times 2},1)^T,\quad \alpha_1(0)+\beta_1(0)=(0_{1\times 2},0)^T.
\]
Then we get
\be\label{Phi2albtrsu}
\ba{rl}
\left(\ba{l}\Phi_{j3}\\\Phi_{33}\ea\right)=&\left(\ba{l}0_{1\times 2}\\1\ea\right)+\frac{1}{k}\left(\ba{l}\Phi_{j3}^{(1)}\\\Phi_{33}^{(1)}\ea\right)+\frac{1}{k^2}\left(\ba{l}\Phi_{j3}^{(2)}\\\Phi_{33}^{(2)}\ea\right)+\cdots\\
&+\left[-\frac{1}{k}\left(\ba{l}\Phi_{j3}^{(1)}(0)\\0\ea\right)+\cdots\right]e^{-8ik^3t}
\ea
\ee
From the first column of the equation (\ref{Phit}) we get
\be\label{Phi1t}
\left\{
\ba{l}
\ba{rl}\Phi_{2\times 2t}=&4k^2U^T\Phi_{3j}+2ik(U^T\bar U\Phi_{2\times 2}+U_x^T\Phi_{3j})\\&-4|u|^2U^T\Phi_{3j}-U_{xx}^T\Phi_{3j}+(u\bar u_x-u_x\bar u)\sig_3\Phi_{2\times 2}\ea\\
\Phi_{3jt}-8ik^3\Phi_{3j}=-4k^2\bar U\Phi_{2\times 2}+2ik(\bar U_x\Phi_{2\times 2}-2|u|^2\Phi_{3j})+4|u|^2\bar U\Phi_{2\times 2}+\bar U_{xx}\Phi_{2\times 2}.
\ea
\right.
\ee
Suppose
\be\label{Phi1xinu}
\left(\ba{l}\Phi_{2\times 2}\\\Phi_{3j}\ea\right)=(\xi_0(t)+\frac{\xi_1(t)}{k}+\frac{\xi_2(t)}{k^2}+\cdots)+(\nu_0(t)+\frac{\nu_1(t)}{k}+\frac{\nu_2(t)}{k^2}+\cdots)e^{8ik^3t}
\ee
where the coefficients $\xi_l(t)$ and $\nu_l(t)$, $l\ge 0$, are independent of $k$. To determine these coefficients,we substitute the above equation into equation (\ref{Phi1t}) and use the initial conditions
\[
\xi_0(0)+\nu_0(0)=(\id_{2\times 2},0_{2\times 1})^T,
\]
Then we get
\be\label{Phi1xinursu}
\ba{rl}
\left(\ba{l}\Phi_{2\times 2}\\\Phi_{3j}\ea\right)=&\left(\ba{l}\id_{2\times 2}\\0_{2\times 1}\ea\right)+\frac{1}{k}\left(\ba{l}\Phi_{2\times 2}^{(1)}\\\Phi_{3j}^{(1)}\ea\right)+\cdots\\
&+\left[\frac{1}{k^2}\left(\ba{l}0\\\nu_2^{(2)}\ea\right)+\cdots\right]e^{8ik^3t}
\ea
\ee
So, from the equation (\ref{cjdef}) and the asymptotic of $s_{j3}(k)$ and $s_{33}(k)$, we get the asymptotic behavior of $c_j(t,k)$ as $k\rightarrow \infty$,
\be\label{cjlargk}
c_j(t,k)=\frac{\Phi_{j3}^{(1)}}{k}+\frac{\Phi_{j3}^{(2)}}{k^2}+\frac{\Phi_{j3}^{(3)}}{k^3}+\cdots.
\ee


\begin{thebibliography}{XXXX}

\bibitem{f1} A. S. Fokas, {\em A unified transform method for solving linear and certain nonlinear PDEs}, Proc. R. Soc. Lond. A {\bf
    453}(1997), 1411-1443.

\bibitem{f2} A. S. Fokas, {\em Integrable nonlinear evolution equations on the half-line}, Commun. Math. Phys. {\bf 230}(2002), 1-39.

\bibitem{f3} A.S. Fokas, {\em A Unified Approach to Boundary Value Problems}, in: CBMS-NSF Regional Conference Series in Applied Mathematics, SIAM, 2008.

\bibitem{abmfs1} A. Boutet De Monvel, A.S. Fokas, D. Shepelsky, {\em Integrable nonlinear evolution equations on a finite interval}, Comm. Math. Phys. {\bf 263} (2006) 133¨C172.

\bibitem{abmfs2} A. Boutet de Monvel,A.S.Fokas,D.Shepelsky, {\em The mKDV equation on the half-line}, J. Inst. Math. Jussieu.{\bf 3}(2004), 139-164.

\bibitem{fis} A. S. Fokas, A. R. Its and L. Y. Sung, {\em The nonlinear Schr\"odinger equation on the half-line}, Nonlinearity. {\bf  18}(2005), 1771-1822.

\bibitem{k} S. Kamvissis, {\em Semiclassical nonlinear Schrödinger on the half line}, J. Math. Phys. {\bf 44} (2003) 5849–5868.

\bibitem{l1} J. Lenells, {\em Boundary value problems for the stationary axisymmetric Einstein equations: a disk rotating around a black hole}, Comm. Math. Phys. {\bf 304} (2011) 585-635.
\bibitem{l2} J. Lenells, A.S. Fokas, {\em Boundary-value problems for the stationary axisymmetric Einstein equations: a rotating disc}, Nonlinearity {\bf 24} (2011) 177-206.

\bibitem{ss} N. Sasa, J. Satsuma, {\em New-type of soliton solutions for a higher-order nonlinear Schrödinger equation}, J. Phys. Soc. Japan {\bf 60} (1991) 409–417.

\bibitem{k} D.J. Kaup, {\em On the inverse scattering problem for cubic eigenvalue problems of the class $\psi_{xxx}+6Q\psi_x+6R\psi=\lam\psi$}, Stud. Appl. Math. {\bf 62} (1980) 189-216.

\bibitem{l3} J. Lenells, {\em Initial-boundary value problems for integrable evolution equations with $3\times 3$ Lax pairs}, Physica D {\bf 241}(2012) 857-875.

\bibitem{l4} J. Lenells, {\em The Degasperis-Procesi equation on the half-line}, Nonlinear Analysis {\bf 76}(2013) 122-139.

\bibitem{bc} R. Beals and R. Coifman, {\em Scattering and inverse scattering for first order systems},  Comm. in Pure and Applied Math. {\bf 37}(1984), 39--90.

\bibitem{dz} P. Deift and  X. Zhou, {\em A steepest descent method for oscillatory Riemann--Hilbert problems},  Ann. of Math. (2) {\bf 137}(1993), 295-368.

\bibitem{pdl} P. D. Lax, {\em Integrals of nonlinear equations of evolution and solitary waves}, Comm. Pure. Appl. Math.{\bf 21}(1968), 467-490.

\bibitem{fl1} A. S. Forkas and J. Lenells, {\em The unified method: \Rmnum{1}.nonlinearizable problem on the half-line}, J. Phys. A: Math. Theor. {\bf 45}(2012) 195201;

\bibitem{fl2} J. Lenells and A. S. Forkas, {\em The unified method: \Rmnum{2}. NLS on the half-line t-periodic boundary conditions}, J. Phys. A: Math. Theor. {\bf 45}(2012) 195202;

\bibitem{fl3} J. Lenells and A. S. Forkas, {\em The unified method: \Rmnum{3}. Nonlinearizable problem on the interval}, J. Phys. A: Math. Theor. {\bf 45}(2012) 195203;

\end{thebibliography}
\end{document}